\documentclass[12pt,authoryear,review]{elsarticle}

\makeatletter

\def\ps@pprintTitle{%
	\let\@oddhead\@empty
	\let\@evenhead\@empty
	\def\@oddfoot{\centerline{\thepage}}%
	\let\@evenfoot\@oddfoot}
\makeatother

\usepackage{amsmath,amssymb,amsfonts}
\usepackage{graphicx}
\usepackage{subfigure}
\usepackage{dsfont}
\usepackage{color}
\usepackage[rightcaption]{sidecap}
\usepackage{epstopdf}
\usepackage{todonotes}

\newtheorem{theorem}{Theorem}
\newtheorem{lemma}[theorem]{Lemma}
\newtheorem{proposition}[theorem]{Proposition}

\newtheorem{assumption}[theorem]{Assumption}
\newenvironment{proof}[1][Proof]{\begin{trivlist}
		\item[\hskip \labelsep {\bfseries #1}]}{\end{trivlist}}

\journal{TBA}

\begin{document}
	\begin{frontmatter}
		
		\title{Optimal Liquidation of Perpetual Contracts}
		
		\author[author1]{Ryan Donnelly}\ead{ryan.f.donnelly@kcl.ac.uk}
		\author[author2]{Junhan Lin}\ead{junhan.1.lin@kcl.ac.uk}
		\author[author3]{Matthew Lorig}\ead{mlorig@uw.edu}
		\address[author1] {King's College London, United Kingdom}
		\address[author2] {King's College London, United Kingdom}
		\address[author3] {University of Washington, Seattle, WA}
		
		\date{}
		
		\begin{abstract}
			An agent holds a position in a perpetual contract with payoff function $\psi$ and attempts to liquidate the position while managing transaction costs, inventory risk, and funding rate payments. By solving the agent's stochastic control problem we obtain a closed-form expression for the optimal trading strategy when the payoff function is given by $\psi(s) = s$. When the payoff function is non-linear we provide approximations to the optimal strategy which apply when the funding rate parameter is small or when the length of the trading interval is small. We further prove that when $\psi$ is non-linear, the short time approximation can be written in terms of the closed-form trading strategy corresponding to the case of the ideneity payoff function.
		\end{abstract}
	\end{frontmatter}
	
	\section{Introduction}
	
	In this paper we investigate how an agent optimally liquidates a position in a perpetual contract before some fixed maturity date. The challenge facing the agent is to determine the optimal trading strategy whilst her trades are subject to market impact, risk associated with price changes of the inventory she is holding, and the cashflow payments which are made between parties that have a non-zero position in the perpetual contract.
	
	Our model captures two distinct components to market impact: temporary price impact, which refers to the immediate effect on the transaction price as a trade consumes liquidity and penetrates through the available orders in the limit order book (LOB), and permanent price impact, which constitutes a long lasting persistent shift in the asset’s mid-price that subsequently affects the transaction prices of all future trades. Previous research on LOB structures and market impact can be found in for example \cite{eisler2012price}, \cite{cont2014price}, and \cite{xu2018multi}. Our research bridges the literature on perpetual contracts with that of the optimal execution problem. Liquidation of large inventory with market impacts has developed from the early models of \cite{bertsimas1998optimal} and \cite{almgren2001optimal} to more recent contributions such as \cite{cartea2016incorporating}, \cite{horst2022portfolio}, and \cite{fouque2022optimal}.
	
	A perpetual contract (sometimes referred to as a perpetual future or perpetual swap) is a financial derivative that gives exposure to an underlying asset without owning the asset itself. This exposure occurs through the exchange of cash flows over time between the long and short positions. The magnitude and direction of this cash flow, called the {\it funding rate}, depends on the price of the underlying asset and the price of the perpetual contract itself. Perpetual contracts are traded extremely actively in cryptocurrency markets, with daily turnovers measured in the billions of USD, so transaction costs and market impact are economically significant. Hence, optimal trading of perpetual contracts is crucial for agents who seek to liquidate their large positions. Trading decisions must balance immediate market impact costs, long-lasting price impact, ongoing funding payments, inventory risk control as well as a terminal liquidation penalty.
	
	Previous work which studies perpetual contracts has largely been related to pricing and hedging. In \cite{angeris2023primer}, model-free expressions for the funding rate together with replication strategies are derived. In \cite{ackerer2025perpetual} the authors derive no-arbitrage pricing formulas for several types of perpetual contracts including linear, inverse, and quanto contracts. Along similar lines, \cite{he2022fundamentals} and \cite{dai2025arbitrage} introduce no-arbitrage bounds for perpetual contract prices, the former including the effects of transaction costs and the latter further incorporating the popular clamping function on the funding mechanism. Most of the existing research regarding perpetual contracts focuses on pricing and hedging with little work having been conducted in the context of optimal liquidation.
	
	In this work we divide the analysis into two sections, one in which the funding rate depends linearly on the spot price, and one where the exposure is an arbitrary function.\footnote{A perpetual contract with non-linear dependence of the funding rate is often called an everlasting option, see \cite{bankmanfried2021everlasting}.} When the funding rate is a linear function of spot, we classify the agent's value function in terms of the solution to a system of ordinary differential equations (ODEs) (Proposition \ref{prop:value_function}) and solve for the optimal trading strategy in closed form (Theorem \ref{prop:optimal_control}). The explicit form of the solution allows us to see directly how the trading strategy depends on remaining inventory and the current funding rate. When the payoff function is non-linear we derive multiple trading strategies which are asymptotically optimal with respect to certain model parameters. The first applies when the funding rate parameter is small
	and we observe that this approximation arises from a perturbation of the Almgren–Chriss optimal strategy (Theorem \ref{prop:approx_nu}). The next two approximations (Theorem \ref{prop:T_approx_nu} and Theorem \ref{prop:closed_form_approx}) apply when the time horizon is short, and we demonstrate the effectiveness of these strategies compared to the Almgren-Chriss strategy for different payoff functions.
	
	
	The remainder of the paper is structured as follows. In Section 2 we propose a trading model for the perpetual contract and formulate an optimal stochastic control problem faced by the agent. In Section 3 we obtain an optimal trading strategy in closed form when the payoff function is the identity function and conduct some analysis of the optimal strategy. In Section 4 we consider an arbitrary payoff function and compute various approximations to the optimal strategy when the funding rate parameter or the length of trading interval are small. We also compare the performances of different approximations which applicable for short maturity through simulations. Section 5 concludes, and longer proofs are deferred to the appendix.
	
	\section{Model}\label{sec:model}
	
	\subsection{Dynamics}\label{sec:dynamic}
	
	In this section we outline the dynamics of the assets involved in the trading problem which will include price impact effects. Additionally we describe the dynamics of the inventory and cash processes of the agent. Let $T>0$ be finite and represent the length of the trading horizon so that all processes are defined on $[0,T]$. We denote by $S = (S_t)_{t\in[0,T]}$ the value of the underlying spot price which will determine the funding rate of the perpetual contract. We denote by $P^\nu = (P^\nu_t)_{t\in[0,T]}$ the (controlled) midprice of the perpetual contract which can be directly traded by the agent and which is subject to price impact effects of trading. We let $Q^\nu = (Q^\nu_t)_{t\in[0,T]}$ denote the (controlled) inventory that the agent holds in the perpetual contract, and the control $\nu = (\nu_t)_{t\in[0,T]}$ represents the rate at which the agent trades (positive and negative values of $\nu_t$ represent buying and selling of the perpetual contract, respectively). The dynamics of the controlled inventory are
	\begin{align}
		Q^\nu_t &= Q_0 + \int_0^t \nu_u\,du\,,
	\end{align}
	for some given initial inventory $Q_0\in\mathbb{R}$. The spot and perpetual prices are given by
	\begin{align}
		S_t &= S_0 + \int_0^t \sigma\,dW^S_u\,,\label{eqn:S_dynamics}\\
		P^\nu_t &= P_0 + \int_0^t b\,\nu_u\,du + \int_0^t \eta\,dW^P_u\,,\label{eqn:P_dynamics}
	\end{align}
	for given initial prices $S_0,P_0\in\mathbb{R}$, where $W^S = (W^S_t)_{t\in[0,T]}$ and $W^P = (W^P_t)_{t\in[0,T]}$ are Brownian motions with constant correlation $\rho\in(-1,1)$. The term $b\,\nu_u$ with $b\geq0$ constant represents a permanent price impact effect due to the agent's trading of the perpetual contract. These trades will also incur a temporary price impact which is modeled by setting the transaction price process of trades equal to $\widehat{P}^\nu = (\widehat{P}^\nu_t)_{t\in[0,T]}$ which is given by
	\begin{align}
		\widehat{P}^\nu_t &= P^\nu_t + k\,\nu_t\,,\label{eqn:P_hat}
	\end{align}
	for $k>0$ a constant. This transaction price represents the price that the agent pays (receives) per unit of the perpetual contract when buying (selling) at rate $\nu_t$. Trading at a faster rate means the agent engages in transactions at less favourable prices compared to a slower rate. Further discussion of permanent and temporary price impact can be found in \cite{cartea2015algorithmic}.
	
	The cash holdings of the agent are affected by the agent's own trades as well as the funding rate. We assume that the funding rate, equal to $\beta\,(P^\nu_t - \psi(S_t))$, is paid continuously by the long side of the contract to the short side, where $\beta>0$ is a constant and $\psi:\mathbb{R}\rightarrow\mathbb{R}$, referred to as the {\it payoff function}, is continuous.\footnote{The most common payoff function is the identity. When the payoff function is non-linear the perpetual contract is sometimes referred to as an everlasting option, see \cite{bankmanfried2021everlasting} and \cite{ackerer2025perpetual}.} We denote the agent's cash process by $X^\nu = (X^\nu_t)_{t\in[0,T]}$ and set it equal to
	\begin{align}
		X^\nu_t &= X_0 - \int_0^t \widehat{P}^\nu_u\,\nu_u + \beta\,Q^\nu_u\,(P^\nu_u - \psi(S_u))\,du\,, 
	\end{align}
	for a given initial cash value $X_0\in\mathbb{R}$. In many perpetual contracts, the funding rate is further modified by a clamping function so that the associated cash flows never exceeds some value in either the positive or negative direction. We do not consider this added complexity for tractability reasons.
	
	Throughout this work we employ the complete filtered probability space $(\Omega, \mathbb{P}, (\mathcal{F}_t)_{t\in[0,T]})$ where $(\mathcal{F}_t)_{t\in[0,T]}$ is the standard augmentation of the natural filtration generated by $(W^S, W^P)$.

	\subsection{Performance Criterion}\label{sec:performance criterion}
	
	The agent wishes to maximize the expected value of her terminal wealth subject to an inventory risk control and liquidation penalty. When trading according to the strategy $\nu$, her performance is given by
	\begin{align}
		H^\nu(t, x, q, p, s) &= \mathbb{E}_{t,x,q,p,s}\biggl[X_T^\nu + Q_T^\nu\,(P_T^\nu -\alpha\,Q_T^\nu) - \phi\int_t^T (Q^\nu_u)^2\,du\biggr] \,,\label{eqn:performance_criterion}
	\end{align}
	where $\mathbb{E}_{t,x,q,p,s}\left[\cdot\right]$ represents expectation conditional on $X^{\nu}_{t}=x$, $Q^{\nu}_{t}=q$, $P^{\nu}_{t}=p$ and $S_{t}=s$. The term $X_T^\nu$ is the value in her cash account at time $T$ and $Q_T^\nu\,P_T^\nu$ is the mark to market value of her remaining inventory. The term $\alpha\,(Q_T^\nu)^2$ with $\alpha >0$ constant represents a penalty of having to liquidate her remaining inventory. Finally, $\phi \geq 0$ acts as a risk control term which penalizes holding large amounts of inventory for long periods of time.
	
	The agent's value function is given by
	\begin{align}
		H(t,x,q,p,s) &= \sup_{\nu \in \mathcal{A}} H^\nu(t,x,q,p,s)\,,\label{eqn:value_function}
	\end{align}
	where the set of admissible trading strategies is
	\begin{align}
		\mathcal{A} &= \biggl\{\nu: \nu \mbox{ is } (\mathcal{F}_t)_{t\in[0,T]}\mbox{-adapted and } \mathbb{E}\biggl[\int_0^T\nu_t^2\,dt\biggr] < \infty\biggr\}\,.
	\end{align}
	
	The control problem posed in \eqref{eqn:value_function} has the associated Hamilton-Jacobi-Bellman (HJB) partial differential equation (PDE):
	\begin{align}
		\partial_{t}H+\sup_{\nu}\left\{\mathcal{L}^{\nu}H\right\}-\phi\, q^{2}=0\,, \qquad H\left(T,x,q,p,s\right)=x+q\,\left(p-\alpha\, q\right)\,,\label{eqn:HJB}
	\end{align}
	where the operator $\mathcal{L}^{\nu}$ is given by
	\begin{align}
		\mathcal{L}^{\nu}=-\biggl((p+k\,\nu)\,\nu+\beta\, q\,(p-\psi\left(s\right))\biggr)\,\partial_{x}+\nu\,\partial_{q}+b\,\nu\,\partial_{p}+\frac{1}{2}\,\sigma^{2}\,\partial_{ss}+\frac{1}{2}\,\eta^{2}\,\partial_{pp}+\rho\,\sigma\,\eta\,\partial_{sp}\,.\label{eqn:L_operator}
	\end{align}

	\section{Identity Payoff Function}\label{sec:identity_function}
	
	In this section we consider the special case of payoff function $\psi(s) = s$ and derive the optimal trading strategy in closed form. To this end, it is convenient to introduce the process $Z = (Z_t^\nu)_{t\in[0,T]}$ defined by $Z^\nu_t = P_t^\nu - S_t$ along with an associated state variable $z = p-s$. Additionally, we assume $2\,\alpha > b$ which ensures that solutions to ODEs appearing in subsequent results do not blow up.
	
	\begin{proposition}[Value Function for Identity Payoff Function]\label{prop:value_function} 
		Suppose $\psi(s) = s$ and define the constant $\Sigma$ by $\Sigma^{2}=\sigma^{2}+\eta^{2}-2\,\rho\,\sigma\,\eta$. Suppose the functions $h_0$, $h_1$, $h_2$, and $h_3$ satisfy the system of ODEs
		\begin{align}\label{eqn:ode}
			\begin{split}
				h'_{0}+\Sigma^{2}\,h_{2}&=0\,,\qquad h_{0}(T)=0\,,\\
				h'_{1} -\phi +\frac{1}{4k}\,\biggl(b\,(1+h_{3})+2\,h_{1}\biggr)^{2}&=0\,, \qquad h_{1}(T)=-\alpha\,,\\
				h'_{2}+\frac{1}{4k}\,(2\,b\,h_{2}+h_{3})^{2}&=0\,, \qquad h_{2}(T)=0\,,\\
				h'_{3}-\beta +\frac{1}{2k}\,(b\,(1+h_{3})+2\,h_{1})\,(2\,b\,h_{2}+h_{3})&=0\,, \qquad h_{3}(T)=0\,.
			\end{split}
		\end{align}
		Then the solution to \eqref{eqn:HJB} is
		\begin{align}
			H(t,x,q,p,s)&=x+q\,p+h(t,q,p-s)\,,\label{eqn:propH}\\
			h(t,q,z)&=h_{0}(t)+h_{1}(t)\,q^{2}+h_{2}(t)\,z^{2}+h_{3}(t)\,q\,z\,.
		\end{align}
	\end{proposition}
	\begin{proof}
		Assuming \eqref{eqn:ode} holds, \eqref{eqn:propH} can be seen to solve the HJB equation \eqref{eqn:HJB} by direct substitution.
	\end{proof}
	
	The form of the value function in \eqref{eqn:propH} shows that a dimensional reduction occurs. The excess value function of the agent, $h$, only depends on the two variables $p$ and $s$ through their difference. At the time of writing, we are unable to solve the system of ODEs \eqref{eqn:ode} in closed form, even through the application of symbolic computer algebra systems. However, we are able to compute the optimal trading strategy in closed-form which appears in Theorem \ref{prop:optimal_control}. This allows us to write the solution to \eqref{eqn:ode} in terms of definite integrals of known functions which can be easily computed numerically. 
	
	\begin{theorem}[Optimal Trading Strategy for Identity Payoff Function]\label{prop:optimal_control} 
		Suppose $\psi(s) = s$. Then the optimal trading speed in feedback form is given by
		\begin{align}
			\nu^{*}(t,q,p,s)=\frac{1}{4\,k}\,\left((\xi(t)+\pi(t))\,q+\frac{1}{b}\,(\xi(t)-\pi(t))\,(p-s)\right)\,,\label{eqn:closed-form nu}
		\end{align}
		where the function $\xi$ and $\pi$ are given by
		\begin{align}
			\xi(t) &= a\,\frac{C\,e^{-2\,\omega\,(T-t)}-1}{C\,e^{-2\,\omega\,(T-t)}+1}\,,\\
			\begin{split}
				\pi(t) &= -\frac{4\,k\,\phi\,(Ce^{-\omega\,(T-t)} + 1)\,(1 - e^{-\omega\,(T-t)})}{a\,(C\,e^{-2\,\omega\,(T-t)}+1)}\\
				& \hspace{20mm} + \frac{e^{-\omega\,(T-t)}}{C\,e^{-2\,\omega\,(T-t)}+1}\,(C+1)\,(b-2\,\alpha)\,,
			\end{split}
		\end{align}
		where
		\begin{align}
			a = 2\,\sqrt{k\,(b\,\beta+\phi)}\,, \qquad C=\frac{a+b-2\,\alpha}{a-b+2\,\alpha}\,, \qquad \omega = \frac{a}{2\,k}\,.\label{eqn:omega}
		\end{align}
		Moreover, the solution provided in \eqref{eqn:propH} is indeed the value function as defined in \eqref{eqn:value_function}. 
	\end{theorem}
	
	\begin{proof}
		For a proof see Appendix A.
	\end{proof}
	
	The optimal trading strategy \eqref{eqn:closed-form nu} in Theorem \ref{prop:optimal_control} shows how the trading speed is affected by the remaining inventory $q$ and (scaled) funding rate $z = p-s$ at time $t$. By noting that $\xi(t) < 0$ for all $t$ and inspecting the ODEs for the functions $f$ and $g$ introduced in the proof of Theorem \ref{prop:optimal_control}, we see that the coefficients of $q$ and $z=p-s$ in \eqref{eqn:closed-form nu} are negative for all $t\in[0,T)$. A negative coefficient on $q$ is typical for optimal liquidation problems when the unaffected price of the traded asset (given by \eqref{eqn:P_dynamics}) is a martingale and when impact effects do not outweigh the terminal penalty (ensured by the assumption $2\,\alpha > b$). This is a result of the agent's desire to minimize the risk associated with holding inventory through time and the penalty associated with terminal inventory holdings. A negative coefficient on $z=p-s$ (except at $t=T$ where the coefficient is zero) is due to the agent's desire to decrease the cost of paying the funding rate in a long position or to increase the profit from receiving the funding rate in a short position.

	In Figure \ref{fig:inventory_density} we plot the (normalized) density of the inventory process for each value of $t$ along with the optimal Almgren-Chriss inventory liquidation path which assumes there is no funding rate. This is done for three different values of the initial funding rate which are positive, zero, and negative in the left, middle, and right panels, respectively. Note that when the initial funding rate is zero (middle panel) the agent behaves similar to the Almgren-Chriss strategy on average when early in the trading period, but then ends up holding higher inventory on average before finally speeding up liquidation towards the end of the trading period. This is due to the impact effects of the agents trades on the perpetual price and the resulting change in the funding rate. When the funding rate is zero, the agent is not rewarded or penalized for holding inventory and so she liquidates as normal. Once their initial liquidating trades have impacted the price, the funding rate will tend to be negative and the agent is rewarded by holding positive inventory. Subsequently, when there is little time left until the agent must liquidate, she speeds up her trading because there is little benefit left in receiving the funding rate and she wishes to avoid the terminal liquidation penalty.
	
	\begin{figure}
		\begin{center}
			\includegraphics[trim=140 240 140 240, scale=0.44]{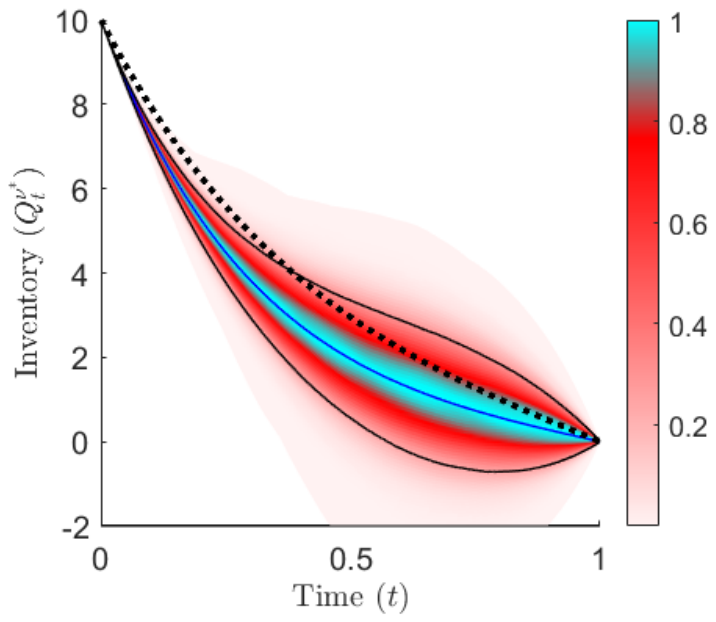}\hspace{1.5mm}
			\includegraphics[trim=140 240 140 240, scale=0.44]{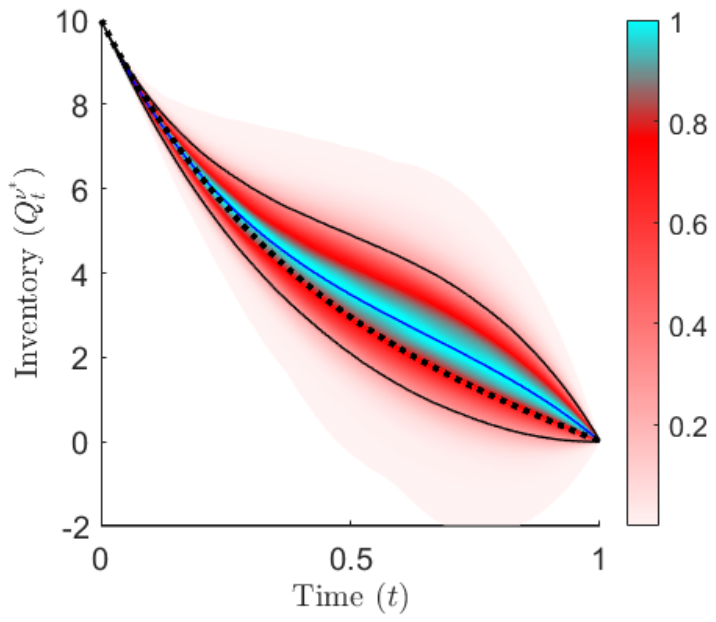}\hspace{1.5mm}
			\includegraphics[trim=140 240 140 240, scale=0.44]{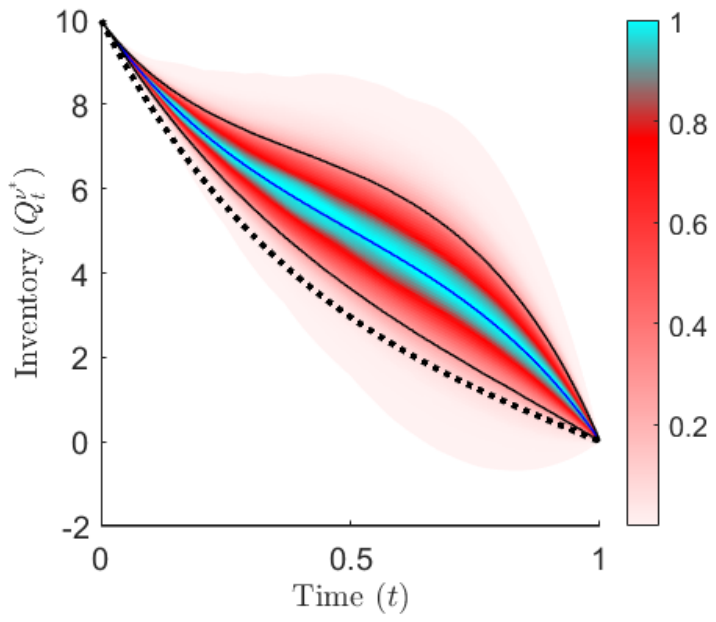}
		\end{center}
		\vspace{-1em}
		\caption{Cross sectional density plots of inventory when trading according to the optimal strategy given in \eqref{eqn:closed-form nu}. The thick dotted curve shows the Almgren-Chriss liquidation strategy. Thin curves represent the $5^{th}$ and $95^{th}$ percentile and the mean. In each panel, the initial spot price is $S_0 = 100$, but the initial perpetual price is $P_0 = 101$ (left), $P_0 = 100$ (middle), and $P_0 = 99$ (right). Parameter values are $T=1$, $k = 0.1$, $b = 0.1$, $\alpha = 100$, $\phi = 0.5$, $\beta = 5$, $\sigma = 1$, $\eta = 1$, $\rho = 0.3$.\label{fig:inventory_density}}
	\end{figure}
	
	In the following proposition we show that when the effect of temporary impact is small, the agent attempts to maintain a relationship between her inventory and the funding rate.
	
	\begin{proposition}\label{prop:behavior}
		Let $\nu^{\ast}$ be the optimal trading strategy for the identity payoff function given in \eqref{eqn:closed-form nu}. Define a stochastic process $A = (A_{t})_{t\in[0,T]}$ by
		\begin{align}
			A_{t} = (b\,\beta + 2\,\phi)\,Q^{\nu^{\ast}}_{t} + \beta\,Z^{\nu^{\ast}}_{t}\,.\label{eqn:A_process}
		\end{align}
		Then the following limit holds
		\begin{align}
			\lim_{k\to 0}\mathbb{E}\biggl[\int_{0}^{T}A_{t}^{2}\,dt\biggr] = 0\,.\label{eqn:prop3}
		\end{align}
	\end{proposition}
	\begin{proof}
		For a proof see Appendix A.
	\end{proof}
	
	Proposition \ref{prop:behavior} gives a rule of thumb that the agent can follow if the market state would not result in significant costs due to trading. Namely, she should trade in such a way that she maintains the process $A$ defined in \eqref{eqn:A_process} to be close to zero. This is similar to other results in portfolio optimization or algorithmic trading in which there is an optimal long-term inventory position which balances risk and return (see for example \cite{cartea2020hedging}). However, after observing the funding rate it is not a direct task of computing the desired inventory which is a multiple of $Z_t$ and submitting the appropriate trade which attains that inventory value, because the trade itself impacts the value of the funding rate.
	
	In Figure \ref{fig:A_paths} we show a sample path of the processes $A$ and $Q^{\nu^*}$ for several values of the temporary price impact parameter $k$. Note that as $k$ decreases the whole path of $A$ tends to become zero (except at times $t=0$ and $t=T$). Indeed, Figure \ref{fig:A_density} shows the cross sectional density of $A$ for three values of $k$ which shows this convergence. The right panel of Figure \ref{fig:A_paths} shows that for moderate values of temporary price impact the inventory tends to ``chase'' the value which is optimal for small impact, but for large values of impact this is too costly to perform.
	
	\begin{figure}
		\begin{center}
			\includegraphics[trim=140 240 140 240, scale=0.55]{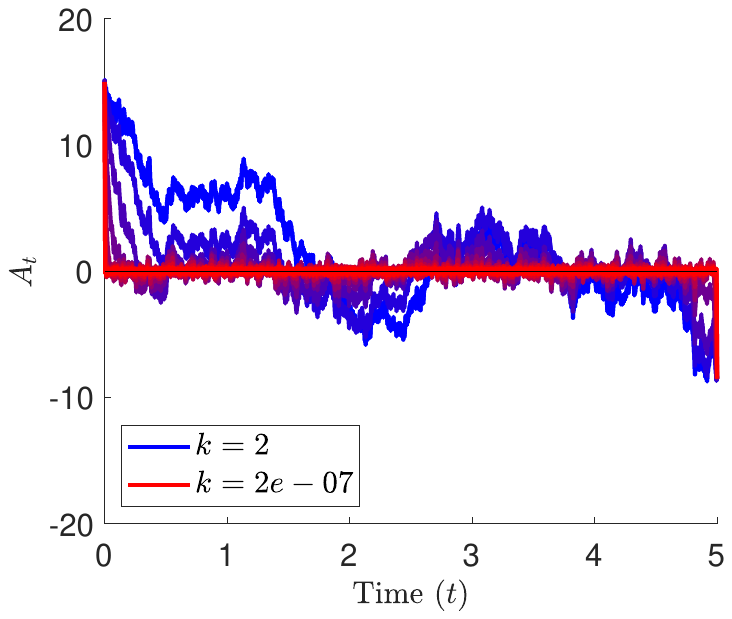}\hspace{15mm}
			\includegraphics[trim=140 240 140 240, scale=0.55]{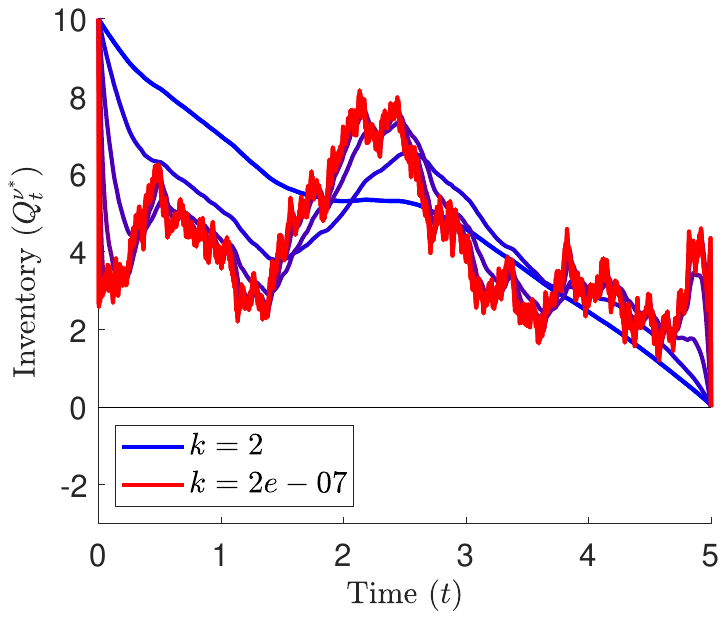}
		\end{center}
		\vspace{-1em}
		\caption{Sample paths of the process $A$ defined in Proposition \ref{prop:behavior} (left panel) and inventory (right panel) for various values of temporary price impact parameter $k$. Other parameter values are $T=5$, $b = 0.1$, $\alpha = 100$, $\phi = 0.5$, $\beta = 5$, $\sigma = 1$, $\eta = 1$, $\rho = 0.3$.\label{fig:A_paths}}
	\end{figure}
	
	\begin{figure}
		\begin{center}
			\includegraphics[trim=140 240 140 240, scale=0.44]{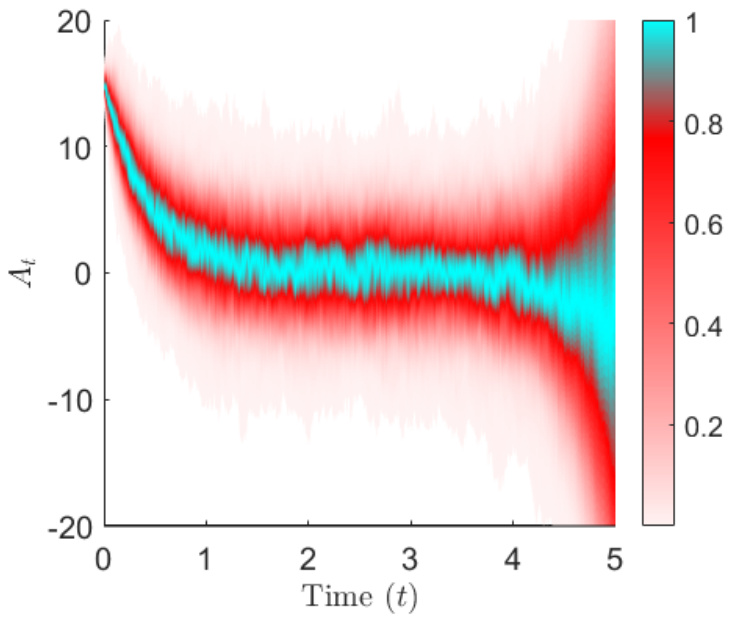}\hspace{2.5mm}
			\includegraphics[trim=140 240 140 240, scale=0.44]{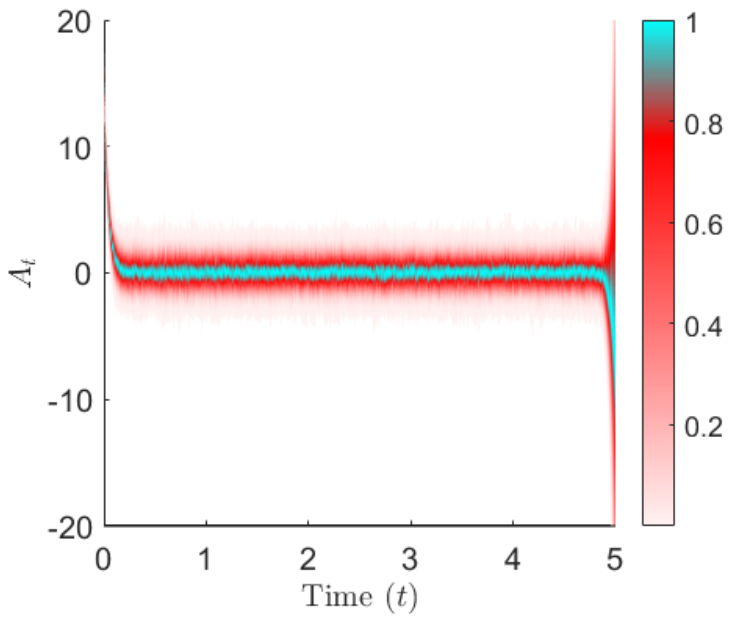}\hspace{2.5mm}
			\includegraphics[trim=140 240 140 240, scale=0.44]{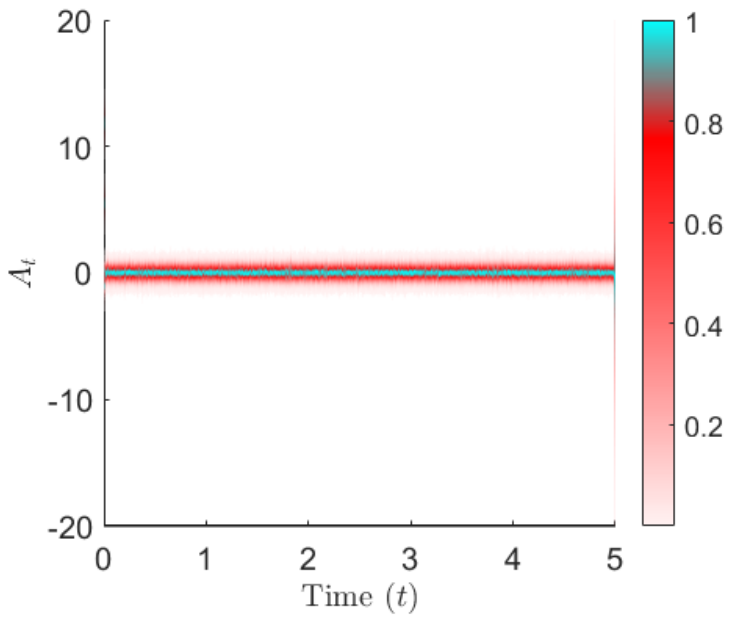}
		\end{center}
		\vspace{-1em}
		\caption{Cross sectional density plots of the process $A$ defined in Proposition \ref{prop:behavior}. The temporary price impact parameter in each panel is $k=2\cdot 10^{-1}$ (left), $k=2\cdot 10^{-3}$ (middle), $k=2\cdot 10^{-5}$ (right). Other parameter values are $T=5$, $b = 0.1$, $\alpha = 100$, $\phi = 0.5$, $\beta = 5$, $\sigma = 1$, $\eta = 1$, $\rho = 0.3$.\label{fig:A_density}}
	\end{figure}

	\section{Arbitrary Payoff Function} \label{sec:arbitrary_function}
	
	In this section we consider the payoff function $\psi$ to be arbitrary with some mild technical restrictions given below. The associated HJB equation \eqref{eqn:HJB} no longer admits the dimensional reduction which appears in \eqref{eqn:propH}, but we still apply the excess value ansatz which takes the form
	\begin{align}\label{eqn:HJB_psi}
		H_\psi(t,x,q,p,s;\beta) &= x + q\,p + h_\psi(t,q,p,s;\beta)\,,
	\end{align}
	where we have emphasized that the value function depends on the payoff function $\psi$ and funding parameter $\beta$. By substitution in \eqref{eqn:HJB} the excess value function $h_\psi$ satisfies
	\begin{align}\label{eqn:HJB_h}
		\begin{split}
			\partial_th_\psi + \frac{1}{2}\,\biggl(\sigma^2\,\partial_{ss}h_\psi + \eta^2\,\partial_{pp}h_\psi + 2\,\rho\,\sigma\,\eta\,\partial_{sp}h_\psi\biggr)  - \phi\,q^2 \hspace{10mm}\\
			- \beta\,q\,(p-\psi(s)) + \sup_\nu \biggl\{-k\,v^2 + (\partial_qh_\psi + b\,(q+\partial_ph_\psi))\,\nu\biggr\} &= 0\,,
		\end{split}
	\end{align}
	with terminal condition $h_\psi(T,q,p,s;\beta) = -\alpha\,q^2$. The supremum in equation \eqref{eqn:HJB_h} is attained at
	\begin{align}
		\nu^* &= \frac{1}{2\,k}\biggl( \partial_qh_\psi + b\,(q + \partial_ph_\psi) \biggr)\,,\label{eqn:nu_star_HJB_h}
	\end{align}
	which upon substitution into \eqref{eqn:HJB_h} gives
	\begin{align}
		\begin{split}
			\partial_th_\psi + \frac{1}{2}\,\biggl(\sigma^2\,\partial_{ss}h_\psi + \eta^2\,\partial_{pp}h_\psi + 2\,\rho\,\sigma\,\eta\,\partial_{sp}h_\psi\biggr) - \phi\,q^2 \hspace{10mm}\\
			- \beta\,q\,(p-\psi(s)) + \frac{1}{4\,k} \biggl(  \partial_qh_\psi + b\,(q + \partial_ph_\psi)   \biggr)^2 &= 0\,.
		\end{split}\label{eqn:HJB_h2}
	\end{align}

	In order to prove the validity of the expansion given below, we make the following technical assumptions.
	\begin{assumption}\label{ass:ass}
		\begin{enumerate}[i)]
			\item $\psi\in C^{4}\left(\mathbb{R}\right)$ with all derivatives bounded.
			\item Given initial states $x$, $q$, $p$ and $s$, there exist positive constants $\epsilon^*$, $\beta^*$, and $K$ that satisfy the following uniform boundedness condition: for every $\epsilon\in\left(0,\epsilon^*\right)$ and $\beta\in\left(0,\beta^*\right)$ if $\nu$ is an admissible control such that
			\begin{align*}
				H^{\nu}_\psi\left(0,x,q,p,s;\beta\right)+\epsilon\ge H_{\psi}\left(0,x,q,p,s;\beta\right)\,,
			\end{align*}
			then for every $t\in [0,T]$
			\begin{align*}
				\mathbb{E} [(Q_{t}^{\nu})^2] & \leq K\,.
			\end{align*}
		\end{enumerate}
	\end{assumption}
	
	Part i) of Assumption \ref{ass:ass} is made for technical convenience in proving the asymptotic convergence of our proposed strategies and can likely be weakened to include more payoff functions, but we want to focus on the derivation and interpretation of such strategies rather than classifying their effectiveness in full generality. Likewise, part ii) of Assumption \ref{ass:ass} is a technical assumption which assists in the proofs of our convergence results. The interpretation of this assumption is that the underlying processes satisfy a type of uniform boundedness condition with respect to the control when controls are restricted to being close to optimal. Similar assumptions about boundedness and regularity are made in other works that derive approximations to optimal trading strategies such as in \cite{ekren2019portfolio} and \cite{cartea2020hedging}. This assumption implies a similar boundedness condition for $P_t^\nu$ because price impact is linear, and $S_t$ satisfies is trivially because it does not depend on the control.
	
	The following theorem gives an approximation of the value function which has an error that vanishes to second order with respect to the funding rate parameter $\beta$.
	
	\begin{theorem}[Asymptotic Approximation of Value Function]\label{prop:asymptotic_approximation}
		The excess value function $h_{\psi}$ admits the following approximation:
		\\i) {\bf Expansion}:
		\begin{align}
			h_{\psi}(t,q,p,s;\beta) &= \widehat{h}_\psi(t,q,p,s;\beta) + R(t,q,p,s;\beta)\,,\\
			\widehat{h}_\psi(t,q,p,s;\beta) &= h_{0}(t,q) + \beta\, h_{1,\psi}(t,q,p,s) + \beta^{2}\,h_{2,\psi}(t,q,p,s)\,,
		\end{align}
		such that
		\begin{align}
			\lim_{\beta\downarrow 0}\frac{1}{\beta^{2}}\,R(t,q,p,s;\beta) &= 0,
		\end{align}
		\\ii) {\bf Zero and First Order Terms}: The functions $h_{0}$ and $h_{1,\psi}$ may be taken as
		\begin{align}
			h_{0}(t,q) &= \gamma(t)\,q^{2}\,,\\
			h_{1,\psi}(t,q,p,s) &= \gamma_{0,\psi}(t,s)\,q+\gamma_{1}(t)\,q\,p+\gamma_{2}(t)\,q^{2}\,,\label{eqn:h_1}
		\end{align}
		where the functions $\gamma$, $\gamma_{0,\psi}$, $\gamma_{1}$ and $\gamma_{2}$ are given as
		\begin{align}
			\gamma(t) &= \frac{\widetilde{a}}{2}\,\frac{\widetilde{C}\,e^{-2\,\widetilde{\omega}\,(T-t)} - 1}{\widetilde{C}\,e^{-2\,\widetilde{\omega}\,(T-t)} + 1} - \frac{b}{2}\,,\\
			\gamma_{0,\psi}(t,s) &= \int_t^T \frac{\widetilde{C}\,e^{-\widetilde{\omega}\,(T-u)} + e^{\widetilde{\omega}\,(T-u)}}{\widetilde{C}\,e^{-\widetilde{\omega}\,(T-t)} + e^{\widetilde{\omega}\,(T-t)}}\,\mathbb{E}[\psi(S_u)| S_t=s]\,du\,,\\
			\gamma_1(t) &= \frac{(\widetilde{C}\,e^{-\widetilde{\omega}\,(T-t)} + 1)\,(e^{-\widetilde{\omega}\,(T-t)} - 1)}{\widetilde{\omega}\,(\widetilde{C}\,e^{-2\,\widetilde{\omega}\,(T-t)}+1)}\,,\\
			\begin{split}
				\gamma_2(t) &= \frac{-b\,e^{-2\,\widetilde{\omega}\,(T-t)}}{2\,\widetilde{\omega}\,(\widetilde{C}\,e^{-2\,\widetilde{\omega}\,(T-t)}+1)^2} \, \biggl(4\,\widetilde{\omega}\,\widetilde{C}\,(T-t) - 2\,(1-\widetilde{C})\,(1-e^{\widetilde{\omega}\,(T-t)}) \\
				& \hspace{10mm} + 2(\widetilde{C}^2-\widetilde{C})\,(1-e^{-\widetilde{\omega}\,(T-t)}) + (1-e^{2\,\widetilde{\omega}\,(T-t)}) - \widetilde{C}^2\,(1-e^{-2\,\widetilde{\omega}\,(T-t)}) \biggr)  \,,
			\end{split}
		\end{align}
		where
		\begin{align}
			\widetilde{a} = 2\,\sqrt{k\,\phi}\,, \qquad \widetilde{C}=\frac{\widetilde{a}+b-2\,\alpha}{\widetilde{a}-b+2\,\alpha}\,, \qquad \widetilde{\omega} = \frac{\widetilde{a}}{2\,k}\,.
		\end{align}
		iii) {\bf Second Order Terms}: The function $h_{2,\psi}$ may be taken as
		\begin{align}
			h_{2,\psi}(t,q,p,s) &= \lambda_{0}(t,s) + \lambda_{1}(t,s)\,q + \lambda_{2}(t)\,q^{2} + \lambda_{3}(t)q\,p + \lambda_{4}(t,s)\,p + \lambda_{5}(t)\,p^{2}\,,\label{eqn:h_2}
		\end{align}
		where $\lambda_{0}$ has at most quadratic growth in $s$, and $\lambda_{1}$ and $\lambda_{4}$ have at most linear growth in $s$.
	\end{theorem}
	
	\begin{proof}
		For a proof see Appendix B.
	\end{proof}
	
	With an approximation to the value function in hand through Theorem \ref{prop:asymptotic_approximation}, one can substitute this approximation into the candidate feedback control \eqref{eqn:nu_star_HJB_h}, which is well defined because it is continuously differentiable, and collect terms according to powers of $\beta$. The following theorem indicates the result of the computation and shows that truncating after terms of order greater than one in $\beta$ results in performance which is accurate to second order.

	\begin{theorem}[Asymptotic Approximation of Optimal Trading Speed]\label{prop:approx_nu}
		Let $\widehat{\nu}$ be a feedback control given by
		\begin{align}
			\widehat{\nu}(t,q,p,s;\beta) &= \nu_{0}(t,q) + \beta\,\nu_{1}(t,q,p,s)\,,\label{eqn:nu_hat}
		\end{align}
		with
		\begin{align}
			\nu_{0}(t,q) &= \frac{1}{2\,k}\,(b+2\,\gamma(t))\,q\,,\\
			\nu_{1}(t,q,p,s) &= \frac{1}{2\,k}\,\biggl(\gamma_{0,\psi}(t,s) + \gamma_{1}(t)\,p + \bigl(2\,\gamma_{2}(t)+b\,\gamma_{1}(t)\bigr)\,q\biggr)\,.
		\end{align}
		Then $\widehat{\nu}_{t}=\widehat{\nu}\left(t,Q^{\widehat{\nu}}_{t},P^{\widehat{\nu}}_{t},S_{t};\beta\right)$ is an admissible control. Defining $h_\psi^{\widehat{\nu}}$ by the relation
		\begin{align}
			H_\psi^{\widehat{\nu}}\left(t,x,q,p,s;\beta\right)=x+q\,p+h_\psi^{\widehat{\nu}}\left(t,q,p,s;\beta\right)\,,
		\end{align}
		$\widehat{\nu}$ is asymptotically optimal to second order with respect to $\beta$. Specifically
		\begin{align}
			\lim_{\beta\rightarrow 0}\frac{h_{\psi}\left(t,q,p,s;\beta\right) - h_\psi^{\widehat{\nu}}\left(t,q,p,s;\beta\right)}{\beta^2} &= 0\,.
		\end{align}
	\end{theorem}
	\begin{proof}
		For the proof see Appendix B.
	\end{proof}
	
	Inspection of the strategy in \eqref{eqn:nu_hat} and comparison to other results in optimal execution give an interpretation for its structure. The term $\nu_0(t,q)$ representing the order zero contribution is the Almgren-Chriss strategy, which is to be expected since we are considering an expansion with respect to the funding parameter $\beta$. The first order correction contains two contributions. The first is $\frac{1}{2\,k}\,(\gamma_{0,\psi}(t,s) + \gamma_{1}(t)\,p)$ which satisfies
	\begin{align*}
		\frac{1}{2\,k}\,\gamma_{0,\psi}(t,s) + \gamma_{1}(t)\,p &= -\frac{1}{2\,k}\,\int_t^T \frac{\widetilde{C}\,e^{-\widetilde{\omega}\,(T-u)} + e^{\widetilde{\omega}\,(T-u)}}{\widetilde{C}\,e^{-\widetilde{\omega}\,(T-t)} + e^{\widetilde{\omega}\,(T-t)}}\,\mathbb{E}[p-\psi(S_u)| S_t=s]\,du\,.
	\end{align*}
	This has an analogous form to execution strategies with an alpha signal, where the signal here is the quantity $p-\psi(s)$ (see for example \cite{cartea2016incorporating} and \cite{neuman2022optimal}). The remaining term $\frac{1}{2\,k}\,\bigl(2\,\gamma_{2}(t)+b\,\gamma_{1}(t)\bigr)\,q$ represents how the agent unwinds the additional inventory which is acquired by taking advantage of the signal $p-\psi(s)$.
	
	In a similar vein to Theorem \ref{prop:asymptotic_approximation}, the following result gives an approximation of the value function when the time remaining until maturity is small.
	
	\begin{theorem}[Asymptotic Approximation of Value Function]\label{prop:T_approx}
		The excess value function $h_{\psi}$ admits the following approximation:
		\begin{align}
			h_{\psi}(t,q,p,s;T) &= \widetilde{h}_\psi(t,q,p,s;T) + \widetilde{R}(t,q,p,s;T)\,,\\
			\widetilde{h}_\psi(t,q,p,s;T) &= \widetilde{h}_{0}(q) + (T-t)\, \widetilde{h}_{1,\psi}(q,p,s) + (T-t)^{2}\,\widetilde{h}_{2,\psi}(q,p,s)\,,
		\end{align}
		such that
		\begin{align}
			\lim_{T\downarrow 0}\frac{1}{T^{2}}\,\widetilde{R}(t,q,p,s;T) &= 0,
		\end{align}
		where the function $\widetilde{h}_{0}$, $\widetilde{h}_{1,\psi}$ and $\widetilde{h}_{2,\psi}$ are given as
		\begin{align}
			\widetilde{h}_{0}(q) &= -\alpha\,q^{2},\\
			\widetilde{h}_{1,\psi}(q,p,s) &= \biggl(\frac{(b - 2 \alpha)^{2}}{4\,k} - \phi\biggr)\,q^{2} - \beta\,(p - \psi(s))\,q,\\
			\widetilde{h}_{2,\psi}(q,p,s) &= \frac{b - 2\,\alpha}{4\,k}\biggl(\frac{(b - 2\,\alpha)^{2}}{2\,k} - 2\,\phi - b\,\beta\biggr)\,q^{2} + \frac{\beta}{4}\,\biggl(-\frac{b - 2\,\alpha}{k}\,(p - \psi(s)) + \sigma^{2}\,\psi'''(s)\biggr)\,q.
		\end{align}
	\end{theorem}
	\begin{proof}
		For a proof see Appendix B.
	\end{proof}
	
	Using a similar process to computing a trading strategy which is approximately optimal as in Theorem \ref{prop:approx_nu}, the approximation to the value function can be substituted into the candidate feedback control \eqref{eqn:nu_star_HJB_h}. The following theorem shows that by truncating the resulting expression after the terms which are linear with respect to $T$, the control obtained yields performance which is accurate to second order.
	
	\begin{theorem}[Asymptotic Approximation of Optimal Trading Speed]\label{prop:T_approx_nu}
		Let $\widetilde{\nu}$ be a feedback control given by
		\begin{align}
			\widetilde{\nu}(t,q,p,s;T) &= \widetilde{\nu}_{0}(q) + (T-t)\,\widetilde{\nu}_{1}(q,p,s)\,,\label{eqn:nu_tilde}
		\end{align}
		with
		\begin{align*}
			\widetilde{\nu}_{0}(q) &= \frac{1}{2\,k}\,(b\,q + \partial_{q}\widetilde{h}_{0})\\
			&= -\frac{2\,\alpha - b}{2\,k}\,q\,,\\
			\widetilde{\nu}_{1}(t,q,p,s) &= \frac{1}{2\,k}\,(\partial_{q}\widetilde{h}_{1,\psi} + b\,\partial_{p}\widetilde{h}_{1,\psi})\\
			&= \frac{1}{2\,k}\biggl( \frac{(2\,\alpha - b)^2}{2\,k} - (b\,\beta+2\,\phi) \biggr)\,q - \frac{\beta}{2\,k}\,(p-\psi(s))\,.
		\end{align*}
		Then $\widetilde{\nu}_{t}=\widetilde{\nu}\left(t,Q^{\widetilde{\nu}}_{t},P^{\widetilde{\nu}}_{t},S_{t};T\right)$ is an admissible control. Defining $h_\psi^{\widetilde{\nu}}$ by the relation
		\begin{align}
			H_\psi^{\widetilde{\nu}}\left(t,x,q,p,s;T\right)=x+q\,p+h_\psi^{\widetilde{\nu}}\left(t,q,p,s;T\right)\,,
		\end{align}
		$\widetilde{\nu}$ is asymptotically optimal to second order with respect to $T$. Specifically
		\begin{align}
			\lim_{T\rightarrow 0}\frac{h_{\psi}\left(t,q,p,s;T\right) - h_\psi^{\widetilde{\nu}}\left(t,q,p,s;T\right)}{T^2} &= 0\,.
		\end{align}
	\end{theorem}
	\begin{proof}
		For the proof see Appendix B.
	\end{proof}
	
	The trading strategy given in \eqref{eqn:nu_tilde} has two contributing terms. Notice that the first term given by $-\frac{2\,\alpha-b}{2\,k}\,q$ does not depend on the running inventory penalty $\phi$ or the funding rate parameter $\beta$. This is because those parameters both affect the performance criterion according to a quantity which accumulates over time, but this term represents the limit of an optimal control as the length of the time horizon approaches zero. In fact, any control which is reasonably close to optimal is equal to this value at time $T$ as can be seen from the terminal condition of equation \eqref{eqn:HJB_h} and the feedback from of the candidate optimal strategy given in \eqref{eqn:nu_star_HJB_h}. The remaining term in the control \eqref{eqn:nu_tilde} captures the agent's attempt to minimize the last remaining portion of the running inventory penalty through $-\frac{\phi}{k}\,q$, and to adjust for the final funding payments through $-\frac{\beta}{2\,k}\,(p-\psi(s))$. The remainder of this term represents the agent compensating their strategy  to avoid associated inventory penalties, and a higher order correction to the constant strategy taken at time $T$ as discussed above.
	
	In the next result we show that the optimal trading strategy which is computed in closed form when the function $\psi$ is the identity may be used to attain performance which is approximately optimal for short time horizons in the case of a general payoff function. Recall the feedback form of this strategy is given by a function $\nu^*:[0,T]\times \mathbb{R}^3\to\mathbb{R}$ written in closed form in \eqref{eqn:closed-form nu}. The approximating strategy is attained by substituting the quantity $\psi(s)$ for the fourth argument in place of $s$.
	
	\begin{proposition}[Closed-form Approximation of Optimal Trading Speed]\label{prop:closed_form_approx}
		The following approximation holds locally uniformly in $(t,q,p,s)$:
		\begin{align}
			\nu^{\ast}(t,q,p,\psi(s);T) = \widetilde{\nu}(t,q,p,s;T) + o(T)\,.
		\end{align}
		Let $\overline{\nu}$ be a feedback control given by
		\begin{align}
			\overline{\nu}(t,q,p,s;T) &= \nu^{\ast}(t,q,p,\psi(s);T)\,.\label{eqn:nu_bar}
		\end{align}
		Then $\overline{\nu}_t = \overline{\nu}(t,Q_t^{\overline{\nu}},P_t^{\overline{\nu}},S_t;T)$ is an admissible control. Define $h^{\overline{\nu}}_{\psi}$ by the relation
		\begin{align}
			H^{\overline{\nu}}_{\psi}(t,x,q,p,s;T) = x + q\,p + h^{\overline{\nu}}_{\psi}(t,q,p,s;T)\,.
		\end{align}
		Then $\overline{\nu}$ is asymptotically approximately optimal to second order with respect to $T$. Specifically,
		\begin{align}
			\lim_{T\rightarrow 0}\frac{h_{\psi}\left(t,q,p,s;T\right) - h_\psi^{\overline{\nu}}\left(t,q,p,s;T\right)}{T^2} &= 0\,.
		\end{align}
	\end{proposition}
	\begin{proof}
		For the proof see Appendix B.
	\end{proof}
	
	Given two different approximations to optimal performance for small values of $T$, it is reasonable to ask if one might typically perform better than the other. To this end, we conduct simulations of both strategies given in \eqref{eqn:nu_tilde} and \eqref{eqn:nu_bar}, along with the corresponding Almgren-Chriss strategy which assumes the funding rate is identically zero, and compare their performance for several values of $T$. These simulations are conducted for two different payoff functions shown in Figure \ref{fig:funding_functions}. In the left panel the payoff function is chosen to be
	\begin{align*}
		\psi(S) &= S + \frac{2\,L}{1 + e^{-\kappa\,(S-S_0 - \Delta_S)}}\,,
	\end{align*}
	with $S_0 = 100$, $\Delta_S = -0.1$, $\kappa = 10$, and $L = 1$. In the right panel the payoff function is
	\begin{align*}
		\psi(S) &= S + L\,(S - S_0 - \Delta_S)^2 + \Delta_\psi\,,
	\end{align*}
	with $S_0 = 100$, $\Delta_S = 0.2$, $\Delta_\psi = -2$, and $L = 5$.
	
	\begin{figure}
		\begin{center}
			\includegraphics[trim=140 240 140 240, scale=0.55]{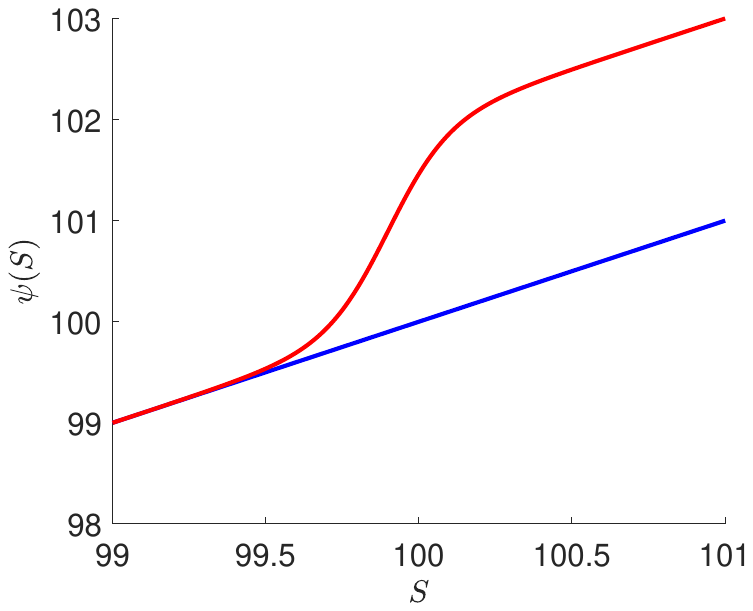}\hspace{15mm}
			\includegraphics[trim=140 240 140 240, scale=0.55]{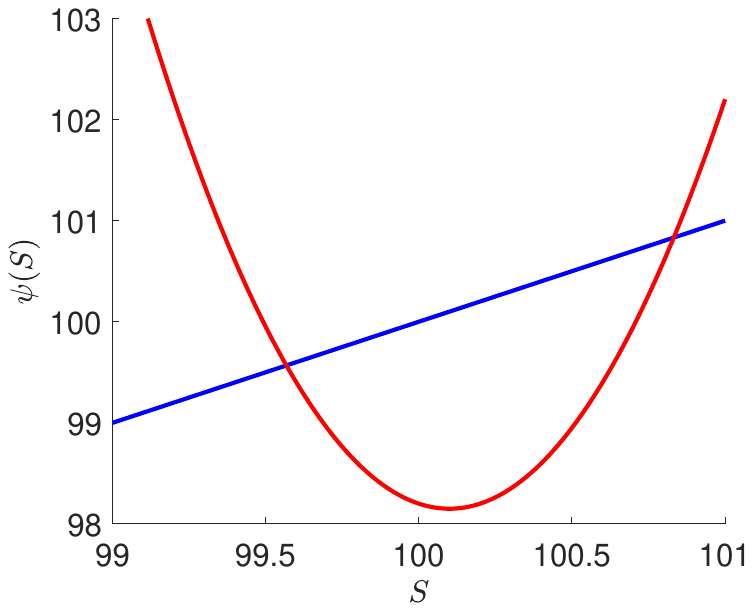}
		\end{center}
		\vspace{-1em}
		\caption{The payoff functions use to demonstrate asymptotic accuracy of trading strategies. The left and right panels add a logistic and quadratic function, respectively, to the identity. \label{fig:funding_functions}}
	\end{figure}
	
	The performance of each strategy applied to both of these payoff functions is shown in Figure \ref{fig:small_T_performance}. Note that as $T$ approaches zero, the excess performance of each strategy approaches $-\alpha\,Q_0^2$. This is to be expected from any reasonable strategy which does not accumulate exorbitant costs due to temporary price impact. For larger values of $T$ in these examples, the performance of $\overline{\nu}$ (blue) is better than that of $\widetilde{\nu}$ (red). While both are approximations to an optimal strategy which applies for small $T$, the superior performance by $\overline{\nu}$ can be explained by the fact that it is derived from a strategy ($\nu^*$ from \eqref{eqn:closed-form nu}) which is optimal for all $T$, albeit for a particular payoff function (identity), and that this strategy is optimal when the funding parameter $\beta$ is equal to zero. Thus, the strategy $\widetilde{\nu}$ tends to deviate from optimality more because it is derived using a method which approximates all elements of the problem under a small $T$ regime. Indeed, as the value of $T$ grows larger, we see in the right panel of Figure \ref{fig:small_T_performance} that the performance of $\widetilde{\nu}$ is substantially worse than that of $\overline{\nu}$, and even worse than the Almgren-Chriss strategy which completely ignores the funding rate.
	
	The two examples in Figure \ref{fig:small_T_performance} show that $h^{\overline{\nu}} > h^{\widetilde{\nu}}$. Through the course of our numerical experiments we find that this is typically the case (generally expected due to the discussion of the previous paragraph) but examples can be found where $h^{\widetilde{\nu}} > h^{\overline{\nu}}$, although this does not hold over a wide range of parameter values. In particular, for larger values of $T$ the strategy $\widetilde{\nu}$ tends to deviate more significantly from optimality.
	
	\begin{figure}
		\begin{center}
			\includegraphics[trim=140 240 140 240, scale=0.55]{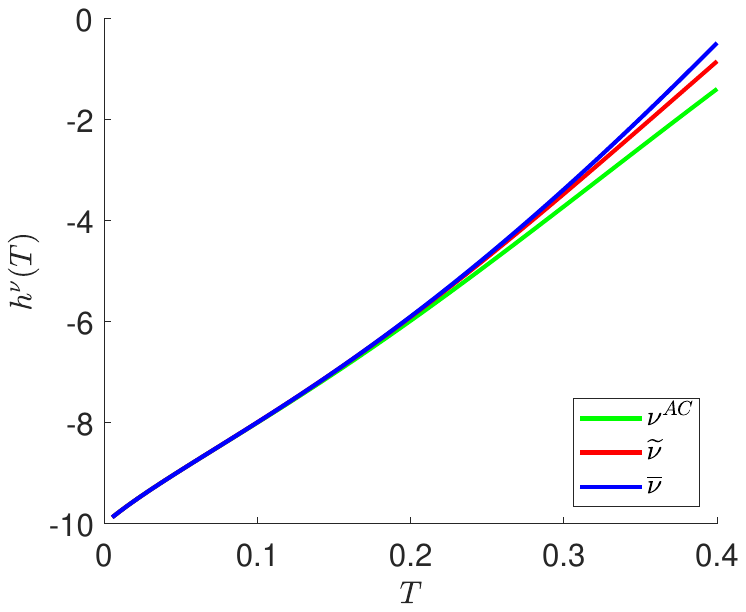}\hspace{15mm}
			\includegraphics[trim=140 240 140 240, scale=0.55]{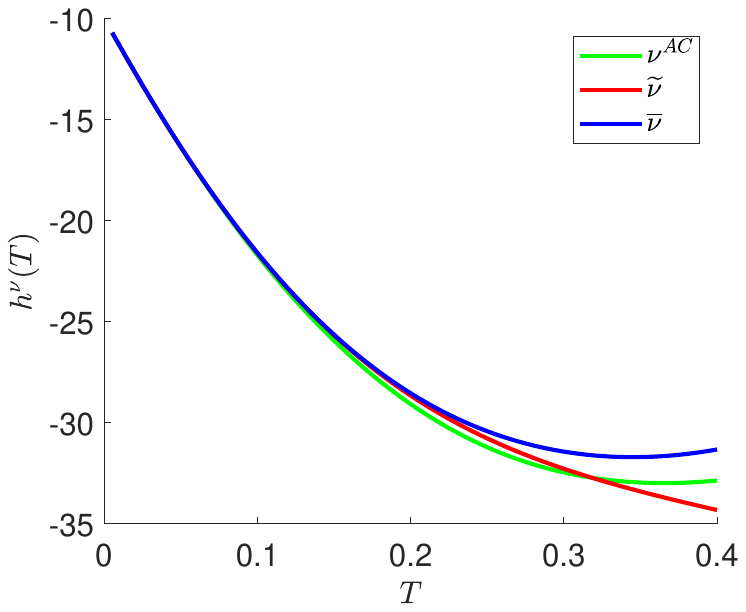}
		\end{center}
		\vspace{-1em}
		\caption{Strategy performance for various values of $T$. The left and right panels use the logistic and quadratic payoff functions, respectively, from Figure \ref{fig:funding_functions}. Other parameter values are $k = 0.1$, $b = 0.1$, $\alpha = 0.1$, $\phi = 0.5$, $\beta = 5$, $\sigma = 1$, $\eta = 1$, $\rho = 0.3$, $Q_0 = 10$, $P_0 = 100$, $S_0 = 100$.\label{fig:small_T_performance}}
	\end{figure}

	\section{Conclusion}\label{sec:conclusion}
	
	We have proposed a model in which an agent is able to trade a perpetual contract written on an underlying spot price process and attempts to maximize expected risk-adjusted terminal wealth when liquidating their position. When the payoff function of the perpetual contract is the identity we solve for the agent's optimal trading strategy in closed form. We derive a limiting relation between inventory and funding rate under small transaction costs. Through simulation studies we demonstrate how the trading pattern deviates from a typical optimal liquidation strategy in the presence of a funding rate, and show that this deviation depends on the initial value of the funding rate. When the payoff function of the perpetual contract is an arbitrary function we propose multiple trading strategies which asymptotically approach optimal performance as either the funding rate parameter or time to maturity vanish. In particular, if one treats the payoff function as the spot price and uses the closed form strategy corresponding to the identity payoff case, then performance is asymptotically optimal for small values of maturity.
	
	\section{Proofs}
	
	\section*{Appendix A: Proofs for Section \ref{sec:identity_function} (Identity Payoff Function)}
	
	
	\begin{proof}[Proof of Theorem \ref{prop:optimal_control}]
		From Proposition \ref{prop:value_function}, the optimizer in the HJB equation is given by
		\begin{align}
			\nu^{*}(t,q,p,s)=\frac{1}{2\,k}\,\biggl(\bigl(2\,h_{1}(t)+b\,(1+h_{3}(t))\bigr)\,q+\bigl(h_{3}(t)+2\,b\,h_{2}(t)\bigr)\,(p-s)\biggr)\,.
		\end{align}
		Define the functions $f$ and $g$ as the coefficients of $q$ and $p-s$, that is
		\begin{align}
			f(t) &= 2\,h_{1}(t)+b\,(1+h_{3}(t))\,,\label{eqn:f}\\
			g(t) &= h_{3}(t)+2\,b\,h_{2}(t)\,.\label{eqn:g}
		\end{align}
		Using \eqref{eqn:ode} we see that $f$ and $g$ satisfy the system of ODEs
		\begin{align}
			f'(t) &= b\,\beta+2\,\phi-\frac{1}{2\,k}\,f(t)\,(b\,g(t)+f(t))\,,\\
			g'(t) &= \beta-\frac{1}{2\,k}\,g(t)\,(b\,g(t)+f(t))\,,
		\end{align}
		with terminal condition $f(T)=b-2\,\alpha$ and $g(T)=0$. We further define $\xi(t)=f(t)+b\,g(t)$ and $\pi(t)=f(t)-b\,g(t)$ which are seen to satisfy
		\begin{align}
			\xi'(t) &= 2\,(b\,\beta+\phi)-\frac{1}{2\,k}\,\xi^{2}(t)\,,\label{eqn:ode_xi}\\
			\pi'(t) &= 2\,\phi-\frac{1}{2\,k}\,\xi(t)\,\pi(t)\,,\label{eqn:ode_pi}
		\end{align}
		with terminal conditions $\xi(T)=\pi(T)=b-2\,\alpha$. The ODE \eqref{eqn:ode_xi} for $\xi$ is uncoupled of Riccati type and has solution
		\begin{align}
			\xi(t) &= a\,\frac{C\,e^{-2\,\omega\,(T-t)}-1}{C\,e^{-2\,\omega\,(T-t)}+1}\,,
		\end{align}
		with $a$, $C$, and $\omega$ as in the statement of the theorem. The ODE \eqref{eqn:ode_pi} for $\pi$ may then be solved directly, and the solution is seen to be 
		\begin{align}
			\begin{split}
				\pi(t) &= -\frac{4\,k\,\phi\,(Ce^{-\omega\,(T-t)} + 1)\,(1 - e^{-\omega\,(T-t)})}{a\,(C\,e^{-2\,\omega\,(T-t)}+1)}\\
				& \hspace{20mm} + \frac{e^{-\omega\,(T-t)}}{C\,e^{-2\,\omega\,(T-t)}+1}\,(C+1)\,(b-2\,\alpha)\,.
			\end{split}
		\end{align}
		The assumption $2\,\alpha > b$ implies $C\in(-1,1)$ which ensures that the expressions for $\xi(t)$ and $\pi(t)$ above are well defined and finite for all $t\in[0,T]$. The definitions of $\xi$ and $\pi$ yield $f(t) = \frac{1}{2}(\xi(t) + \pi(t))$ and $g(t) = \frac{1}{2\,b}(\xi(t) - \pi(t))$, thus the feedback form of the optimal trading strategy is
		\begin{align}
			\nu^{*}(t,q,p,s)=\frac{1}{4\,k}\,\biggl((\xi(t) + \pi(t))\,q+\frac{1}{b}\,(\xi(t) - \pi(t))\,(p-s)\biggr)\,.
		\end{align}
		This control is linear with respect to the state variables with bounded coefficients and therefore is admissible. A standard verification argument shows that the solution to the HJB equation given in Proposition \ref{prop:value_function} is the value function as defined in \eqref{eqn:value_function}.
		\qed
	\end{proof}
	
	\begin{proof}[Proof of Proposition \ref{prop:behavior}]
		Define a stochastic process $Y = (Y_{t})_{t\in[0,T]}$ by
		\begin{align}
			Y_{t} = \frac{1}{\sqrt{k}}\,\biggl(f(t)\,Q^{\nu^{\ast}}_{t} + g(t)\,Z^{\nu^{\ast}}_{t}\biggr)\,,
		\end{align}
		where $f(t) = \frac{1}{2}(\xi(t) + \pi(t))$ and $g(t) = \frac{1}{2\,b}(\xi(t) - \pi(t))$ as in the proof of Theorem \ref{prop:optimal_control}. Application of It\^o's Lemma to the process $Y$ yields
		\begin{align}
			dY_{t} = \frac{1}{\sqrt{k}}\,A_{t}\,dt + \frac{g(t)\,\Sigma}{\sqrt{k}}\, dW^{Z}_{t}\,,
		\end{align}
		where $W^Z = (W^Z_t)_{t\in[0,T]}$ is a Brownian motion defined to satisfy $\Sigma\,dW^Z_{t} = \eta\,dW^P_{t} - \sigma\,dW^S_{t}$. Applying It\^o's Lemma again to the process $A$ yields
		\begin{align}
			dA_{t} = \frac{1}{\sqrt{k}}(b\,\beta + \phi)\, Y_{t}\,dt + \beta\,\Sigma\,dW^Z_{t}\,.
		\end{align}
		Define a 2-dimensional vector process $V = (V_t)_{t\in[0,T]}$ by $V_{t} = [Y_{t}, A_{t}]^{T}$ which has dynamics
		\begin{align}
			dV_{t} = M_{k}\,V_{t}\,dt + u_{k}(t)\,dW^{Z}_{t}\,,
		\end{align}
		where
		\begin{align}
			M_{k} = 
			\begin{bmatrix}
				0 & \frac{1}{\sqrt{k}}\\
				\frac{b\,\beta + \phi}{\sqrt{k}} & 0
			\end{bmatrix}\,,\qquad
			u_{k}(t) = 
			\begin{bmatrix}
				\frac{g(t)\,\Sigma}{\sqrt{k}}\\
				\beta\,\Sigma
			\end{bmatrix}\,.
		\end{align}
		From equation (6.10) in \cite{karatzas1991brownian}, the expectation $V_{t}$ can be written as
		\begin{align}
			\mathbb{E}[V_{t}] = \Phi(t)\,V_{0}\,,
		\end{align}
		where $\Phi$ is the solution of the matrix differential equation
		\begin{align}
			\Phi'(t) = M_{k}\,\Phi(t)\,,\qquad
			\Phi(0) = 
			\begin{bmatrix}
				1 & 0\\
				0 & 1
			\end{bmatrix}\,.
		\end{align}
		This equation has solution
		\begin{align}
			\Phi(t) &= e^{M_{k}\,t}\,\\
			&=\cosh\biggl(\frac{m\,t}{\sqrt{k}}\biggr)\,
			\begin{bmatrix}
				1 & 0\\
				0 & 1
			\end{bmatrix} + \frac{\sqrt{k}}{m}\,\sinh\biggl(\frac{m\,t}{\sqrt{k}}\biggr)\,M_{k}\,,
		\end{align}
		where $m=\sqrt{b\,\beta + \phi}$. Hence, the expectation of $A_{t}$ can be written as
		\begin{align}
			\mathbb{E}[A_{t}] = \cosh(\omega\,t)\, A_{0} + m\,\sinh (\omega\,t)\,Y_{0}\,,
		\end{align}
		with $\omega= \sqrt{\frac{b\beta+\phi}{k}}$ as in \eqref{eqn:omega} of Theorem \ref{prop:optimal_control}. For $t\neq\{0,T\}$ a tedious but direct computation yields
		\begin{align*}
			\lim_{k\rightarrow 0} \mathbb{E}[A_t] &= \left\{ 
			\begin{array}{cl}
				A_0\,, & t=0\\
				0\,, & 0<t<T\\
				-b\,\beta\,Q_0 + \beta\,Z_0\,, & t=T\end{array}\right.\,.
		\end{align*}
		From equation (6.6) in \cite{karatzas1991brownian} and by using the It\^o isometry, the covariance matrix of $V_{t}$ can be written as
		\begin{align*}
			\text{Cov}(V_{t}) &= \text{Cov}\biggl(\Phi(t)\int^{t}_{0}\Phi^{-1}(s)\,u_{k}(s)\,dW^{Z}_{s}\biggr)\\
			&= \int^{t}_{0}\Phi(t-s)\,u_{k}(s)\,(u_{k}(s))^{T}\,(\Phi(t - s))^{T}\,ds.
		\end{align*}
		Let $[\cdot]_{2}$ represent the bottom element of a 2-dimensional vector and let $[\cdot]_{2,2}$ represent the $(2,2)$ entry of a $2\times 2$ matrix. Then the variance of $A_{t}$ is
		\begin{align}
			\text{Var}(A_{t}) = [\text{Cov}(V_{t})]_{2,2} = \int^{t}_{0}\biggl([\Phi(t - s)\,u_{k}(s)]_{2}\biggr)^2\,ds.\label{eqn:A_variance}
		\end{align}
		Another tedious but direct computation gives
		\begin{align*}
			[\Phi(t - s)\,u_{k}(s)]_{2} &= \beta\,\Sigma\,\cosh(\omega\,(t - s)) + \omega\,\Sigma\,g(s)\,\sinh(\omega\,(t - s))\\
			&= \frac{\beta\,\Sigma}{2\,(C\,e^{-2\,\omega\,(T-s)}+1)}\biggl(e^{\omega\,(t-s)}\,e^{-\omega\,(T-s)}\,(2\,C\,e^{-\omega(T-s)} - C+1) \\
			& \hspace{10mm} + e^{-\omega(t-s)}((C-1)e^{-\omega(T-s)}+2)\biggr)\,.
		\end{align*}
		From this expression we see
		\begin{align*}
			\lim_{k\to 0}[\Phi(t - s)\,u_{k}(s)]_{2} &= \left\{ 
			\begin{array}{cl}
				\beta\,\Sigma\,, & s=t \mbox{ or } t=T\\
				0\,, & s<t<T\end{array}\right.\,.
		\end{align*}
		The Dominated Convergence Theorem may be used to interchange the integral and limit in \eqref{eqn:A_variance} which yields
		\begin{align*}
			\lim_{k\rightarrow 0}\text{Var}(A_{t}) &= \left\{ 
			\begin{array}{cl}
				\beta^2\,\Sigma^2\,T\,, & t=T\\
				0\,, & t<T\end{array}\right.\,.
		\end{align*}
		Finally the limit in \eqref{eqn:prop3} holds since
		\begin{align}
			\lim_{k\to 0}\mathbb{E}\biggl[\int_{0}^{T}(A_{t})^{2}\,dt\biggr] = \lim_{k\to 0}\int_{0}^{T}\text{Var}(A_{t}) + \mathbb{E}[A_{t}]^{2}\,dt = 0.
		\end{align}
		The first claim follows from Fubini's Theorem and the second claim follows from Dominated Convergence Theorem.
		\qed
	\end{proof}

	\section*{Appendix B: Proofs for Section \ref{sec:arbitrary_function} (Arbitrary Payoff Function)}
	
	The following two Lemmas are used repeatedly in the proofs of the approximation results which appear in this appendix.
	
	\begin{lemma}\label{lem:lemma}
		Suppose $\psi$ satisfies Assumption \ref{ass:ass} i). For an integrable function $\zeta:\mathbb{R}\to\mathbb{R}$, we define
		\begin{align}
			g(t,s)=\mathbb{E}\biggl[\int^{T}_{t}\zeta(u)\,\psi(S_{u})\,du\biggl| S_{t}=s\biggr]\,,
		\end{align}
		then $g\left(t,s\right)$ is Lipschitz with respect to the variable $s$, uniformly in $t$.
	\end{lemma}
	
	\begin{lemma}\label{lem:lemma2}
		Suppose $\theta:[0,T]\times\mathbb{R}\rightarrow\mathbb{R}$ is continuous with $\partial_s\theta$ continuous and bounded, and suppose $\zeta:[0,T]\rightarrow\mathbb{R}$ is integrable. Define
		\begin{align}
			g_1(t,s) &= \mathbb{E}\biggl[\int_0^T \zeta(u)\,\theta(u,S_u)\,du \biggl|S_t = s\biggr]\,,\\
			g_2(t,s) &= \mathbb{E}\biggl[\int_0^T \zeta(u)\,\theta^2(u,S_u)\,du \biggl|S_t = s\biggr]\,.
		\end{align}
		Then $\partial_sg_1$ is bounded and $\partial_sg_2$ has linear growth in $s$ uniformly in $t$.
	\end{lemma}
	
	\begin{proof}[Proof of Lemma \ref{lem:lemma}]
		From the dynamics of $S$ given in \eqref{eqn:S_dynamics} the transition density of this process between times $t$ and $u$ is
		\begin{align}
			p(z;t,u,s) &= \frac{1}{\sqrt{2\,\pi\,\sigma^{2}\,(u-t)}}\exp\biggl(-\frac{(z-s)^{2}}{2\,\sigma^{2}\,(u-t)}\biggr)\,,
		\end{align}
		By Fubini's Theorem the function $g$ can be written
		\begin{align*}
			g(t,s) &= \int_t^T \zeta(u) \, \mathbb{E}[\psi(S_u)|S_t=s]\,du\\
			&= \int_t^T \zeta(u) \int_\mathbb{R} \psi(z)\,p(z;t,u,s)\,dz\,du\\
			&= \int_t^T \zeta(u) \int_\mathbb{R} \psi(x+s)\,p(x;t,u,0)\,dx\,du\,.
		\end{align*}
		Thus, we have
		\begin{align*}
			|g(t,s_1) - g(t,s_2)| &\leq \int_t^T |\zeta(u)| \, \int_\mathbb{R} |\psi(x+s_1) - \psi(x+s_2)|\,p(x;t,u,0)\,dx \,du\,.
		\end{align*}
		The function $\psi$ is Lipschitz because it has continuous bounded first derivative, therefore
		\begin{align*}
			|g(t,s_1) - g(t,s_2)| &\leq \int_t^T |\zeta(u)| \, \int_\mathbb{R} L_1\,|s_1 - s_2|\,p(x;t,u,0)\,dx \,du\\
			&= L_1\,|s_1-s_2|\,\int_t^T |\zeta(u)|\,du\\
			&\leq L_1\,|s_1-s_2|\,\int_0^T |\zeta(u)|\,du\\
			&= L_2\,|s_1-s_2|\,.
		\end{align*}
		\qed
	\end{proof}
	
	\begin{proof}[Proof of Lemma \ref{lem:lemma2}]
		From the dynamics of $S$ given in \eqref{eqn:S_dynamics} the transition density of this process between times $t$ and $u$ is
		\begin{align}
			p(z;t,u,s) &= \frac{1}{\sqrt{2\,\pi\,\sigma^{2}\,(u-t)}}\exp\biggl(-\frac{(z-s)^{2}}{2\,\sigma^{2}\,(u-t)}\biggr)\,,
		\end{align}
		By Fubini's Theorem the function $g_1$ can be written
		\begin{align*}
			g_1(t,s) &= \int_t^T \zeta(u)\,\mathbb{E}[\theta(u,S_u)|S_t = s]\,du\\
			&= \int_t^T \zeta(u)\, \int_\mathbb{R} \theta(u,z)\,p(z;t,u,s)\,dz  \,du\\
			&= \int_t^T \zeta(u)\, \int_\mathbb{R} \theta(u,x+s)\,p(x;t,u,0)\,dx  \,du\,.
		\end{align*}
		By the Leibniz integration rule, we compute
		\begin{align*}
			\partial_sg_1(t,s) &= \int_t^T \zeta(u)\, \int_\mathbb{R} \partial_s\theta(u,x+s)\,p(x;t,u,0)\,dx  \,du\\
			|\partial_sg_1(t,s)| &\leq \int_t^T |\zeta(u)|\, \int_\mathbb{R} |\partial_s\theta(u,x+s)|\,p(x;t,u,0)\,dx  \,du\\
			&\leq K\int_0^T |\zeta(u)|\,du\,.
		\end{align*}
		Similarly, we compute
		\begin{align*}
			g_2(t,s) &= \int_t^T \zeta(u)\, \int_\mathbb{R} \theta^2(u,x+s)\,p(x;t,u,0)\,dx  \,du\\
			\partial_sg_2(t,s) &= \int_t^T \zeta(u)\, \int_\mathbb{R} 2\,\theta(u,x+s)\partial_s\,\theta(u,x+s)\,p(x;t,u,0)\,dx  \,du\\
			|\partial_sg_2(t,s)| &\leq \int_t^T |\zeta(u)|\, \int_\mathbb{R} 2\,|\theta(u,x+s)||\partial_s\,\theta(u,x+s)|\,p(x;t,u,0)\,dx  \,du\,.
		\end{align*}
		Since $\partial_s\theta$ is continuous and bounded, $\theta$ has linear growth in $s$ uniformly in $t$ and we write
		\begin{align*}
			|\partial_sg_2(t,s)| &\leq K \int_0^T |\zeta(u)|\, \int_\mathbb{R} (1+|x+s|)\,p(x;t,u,0)\,dx  \,du\\
			&\leq K \int_0^T |\zeta(u)|\, \int_\mathbb{R} (1+|x|)\,p(x;t,u,0)\,dx\,du + K\,\int_0^T|\zeta(u)|\,\int_\mathbb{R}|s|\,p(x;t,u,0)\,dx  \,du\\
			&\leq K' (1+|s|)\,.
		\end{align*}
		\qed
	\end{proof}

	\begin{proof}[Proof of Theorem \ref{prop:asymptotic_approximation}]
		\underline{Part I} (formal solution): Substituting $\widehat{h}_\psi$ into the left hand side of \eqref{eqn:HJB_h2} and setting terms proportional to $\beta^{0}$ to vanish gives
		\begin{align}
			\partial_{t}h_{0} -\phi\,q^2 + \frac{1}{4\,k}\,(\partial_{q}h_{0} + b\,q)^{2}=0\,,
		\end{align}
		with terminal condition $h_{0}(T,q) = -\alpha\, q^{2}$. It is easily verified that this equation has solution given by
		\begin{align}
			h_{0}(t,q) &= \gamma(t)\,q^{2}\,,\\
			\gamma(t) &= -\frac{\widetilde{a}}{2}\frac{\widetilde{C}e^{-2\,\widetilde{\omega}\,(T-t)} - 1}{\widetilde{C}\,e^{-2\,\widetilde{\omega}\,(T-t)} + 1} - \frac{b}{2}\,,
		\end{align}
		with $\widetilde{a}$, $\widetilde{C}$, and $\widetilde{\omega}$ as in the statement of the theorem. Similarly, grouping terms proportional to $\beta^{1}$ gives
		\begin{align}
			\begin{split}
				\partial_{t}h_{1,\psi} + \frac{1}{2}\,(\sigma^{2}\,\partial_{ss}h_{1,\psi} + \eta^{2}\,\partial_{pp}h_{1,\psi} + 2\,\rho\,\sigma\,\eta\,\partial_{sp}h_{1,\psi}) - q\,(p-\psi(s))\\
				+ \frac{1}{2\,k}(\partial_{q}h_{1,\psi} + b\,\partial_{p}h_{1,\psi})\,(\partial_qh_0+b\,q) = 0\,,
			\end{split}\label{eqn:h1_PDE}
		\end{align}
		with terminal condition $h_{1,\psi}(T,q,p,s)=0$. We now write $h_{1,\psi}(t,q,p,s)$ in the form $h_{1,\psi}(t,q,p,s) = \gamma_{0,\psi}(t,s)\,q+\gamma_{1}(t)\,q\,p + \gamma_{2}(t)\,q^{2}$, substitute this into \eqref{eqn:h1_PDE} and set the $q$, $q\,p$ and $q^{2}$ terms to vanish independently, obtaining
		\begin{align}
			\partial_{t}\gamma_{0,\psi} + \frac{1}{2}\,\sigma^{2}\,\partial_{ss}\gamma_{0,\psi} + \psi(s)+\frac{1}{2\,k}\,(2\,\gamma + b)\,\gamma_{0,\psi} &= 0\,,\\
			\partial_{t}\gamma_{1} - 1 + \frac{1}{2\,k}\,(2\,\gamma + b)\,\gamma_{1} &= 0\,,\\
			\partial_{t}\gamma_{2} + \frac{1}{2\,k}\,(2\,\gamma + b)\,(b\,\gamma_{1}+2\,\gamma_{2}) &= 0\,,
		\end{align}
		with terminal conditions $\gamma_{0,\psi}(T,s) = \gamma_{1}(T) = \gamma_{2}(T)=0$. The solutions to the ODEs for $\gamma_{1}$ and $\gamma_{2}$ are
		\begin{align}
			\gamma_1(t) &= \frac{(\widetilde{C}\,e^{-\widetilde{\omega}\,(T-t)} + 1)\,(e^{-\widetilde{\omega}\,(T-t)} - 1)}{\widetilde{\omega}\,(\widetilde{C}\,e^{-2\,\widetilde{\omega}\,(T-t)}+1)}\,,\\
			\begin{split}
				\gamma_2(t) &= \frac{-b\,e^{-2\,\widetilde{\omega}\,(T-t)}}{2\,\widetilde{\omega}\,(\widetilde{C}\,e^{-2\,\widetilde{\omega}\,(T-t)}+1)^2} \, \biggl(4\,\widetilde{\omega}\,\widetilde{C}\,(T-t) - 2\,(1-\widetilde{C})\,(1-e^{\widetilde{\omega}\,(T-t)}) \\
				& \hspace{10mm} + 2(\widetilde{C}^2-\widetilde{C})\,(1-e^{-\widetilde{\omega}\,(T-t)}) + (1-e^{2\,\widetilde{\omega}\,(T-t)}) - \widetilde{C}^2\,(1-e^{-2\,\widetilde{\omega}\,(T-t)}) \biggr)  \,,
			\end{split}
		\end{align}
		and by the Feynman-Kac formula, the solution to the PDE of $\gamma_{0,\psi}$ is
		\begin{align}
			\gamma_{0,\psi}(t,s) &= \int_t^T \frac{\widetilde{C}\,e^{-\widetilde{\omega}\,(T-u)} + e^{\widetilde{\omega}\,(T-u)}}{\widetilde{C}\,e^{-\widetilde{\omega}\,(T-t)} + e^{\widetilde{\omega}\,(T-t)}}\,\mathbb{E}[\psi(S_u)| S_t=s]\,du\,.
		\end{align}
		Finally, grouping the terms proportional to $\beta^{2}$ and setting them equal to zero gives
		\begin{align}
			\begin{split}
				\partial_{t}h_{2,\psi} + \frac{1}{2\,k}\,(\partial_{q}h_{2,\psi} + b\,\partial_{p}h_{2,\psi})\,(2\,\gamma + b)\,q + \frac{1}{4\,k}\, (\gamma_{0,\psi} + \gamma_{1}\,p + (b\,\gamma_{1} + 2\,\gamma_{2})\,q)^{2}\\
				+ \frac{1}{2}\,(\sigma^{2}\,\partial_{ss}h_{2,\psi} + \eta^{2}\,\partial_{pp}h_{2,\psi} + 2\,\rho\,\sigma\,\eta\,\partial_{sp}h_{2,\psi}) &= 0\,,
			\end{split}\label{eqn:HJB_h_beta2}
		\end{align}
		with terminal condition $h_{2,\psi}(T,q,p,s)=0$. Writing $h_{2,\psi}$ in the form
		\begin{align}
			h_{2,\psi}(t,q,p,s) &= \lambda_{0}(t,s) + \lambda_{1}(t,s)\,q + \lambda_{2}(t)\,q^{2} + \lambda_{3}(t)\,q\,p + \lambda_{4}(t,s)\,p + \lambda_{5}(t)\,p^{2}\,,
		\end{align}
		substituting into \eqref{eqn:HJB_h_beta2}, and grouping terms by like powers shows that $\{\lambda_{i}\}_{i=0,\dots,5}$ satisfies the system of differential equations
		\begin{align*}
			\partial_t\lambda_0 + \frac{1}{2}\,\sigma^2\,\partial_{ss}\lambda_0 + \rho\,\sigma\,\eta\,\partial_s\lambda_4 + \eta^2\,\lambda_5 + \frac{\gamma_{0,\psi}^2}{4\,k} &= 0\,, & \lambda_0(T,s) &= 0\,,\\
			\partial_t\lambda_1 + \frac{1}{2}\,\sigma^2\,\partial_{ss}\lambda_1 + \frac{2\,\gamma + b}{2\,k}\,\lambda_1 + \frac{b\,(2\,\gamma + b)}{2\,k}\,\lambda_{4} + \frac{(b\,\gamma_1 + 2\,\gamma_2)\,\gamma_{0,\psi}}{2\,k} &= 0\,, & \lambda_1(T,s) &= 0\,,\\
			\lambda_2' + \frac{2\,\gamma + b}{k}\,\lambda_{2} + \frac{b\,(2\,\gamma + b)}{2\,k}\,\lambda_3 + \frac{(b\,\gamma_1 + 2\,\gamma_2)^2}{4\,k} &= 0\,, & \lambda_2(T) &= 0\,,\\
			\lambda_3' + \frac{2\,\gamma + b}{2\,k}\,\lambda_{3} + \frac{b\,(2\,\gamma + b)}{k}\,\lambda_{5} + \frac{(b\,\gamma_1 + 2\,\gamma_2)\,\gamma_1}{2\,k} &= 0\,, & \lambda_3(T) &= 0\,,\\
			\partial_t\lambda_4 + \frac{1}{2}\,\sigma^2\,\partial_{ss}\lambda_4 + \frac{\gamma_{0,\psi}\,\gamma_1}{2\,k} &= 0\,, & \lambda_4(T,s) &= 0\,,\\
			\lambda_5' + \frac{\gamma_1^2}{4\,k} &= 0\,, & \lambda_5(T) &= 0\,.
		\end{align*}
		The solution for each $\lambda_i$ can be written using the Feynman-Kac formula, and then by Lemma \ref{lem:lemma2} we see that $\partial_s\lambda_1$ and $\partial_s\lambda_4$ are continuous and bounded, and thus $\lambda_1$ and $\lambda_4$ have linear growth in $s$ uniformly in $t$. Additionally from Lemma \ref{lem:lemma2}, $\partial_s\lambda_0$ has linear growth in $s$ uniformly in $t$.
		
		\underline{Part II:} (accuracy of approximation).
		With $\widehat{h}_\psi$ as given in the theorem, define
		\begin{align}
			\widehat{H}_\psi(t,x,q,p,s;\beta) &= x + q\,p + \widehat{h}_\psi(t,q,p,s;\beta)\,.
		\end{align}
		For simplicity, we prove the approximation holds for $t=0$ with initial states given by $x$, $q$, $p$, and $s$. The case of $t\ne0$ follows similarly. Let $\nu^{\beta,\epsilon}$ be an admissible control which is $\epsilon\,\beta^{2}$-optimal. Specifically, the control satisfies
		\begin{align}
			H^{\nu^{\beta,\epsilon}}(0,x,q,p,s;\beta) + \epsilon\,\beta^{2} \ge H_{\psi}(0,x,q,p,s;\beta)\,.
		\end{align}
		Define the process $G$ by
		\begin{align}
			G_{t} = \widehat{H}_\psi(t,X_{t}^{\nu^{\beta,\epsilon}},Q_{t}^{\nu^{\beta,\epsilon}},P_{t}^{\nu^{\beta,\epsilon}},S_{t};\beta) - \int_0^t \phi\,(Q_{u}^{\nu^{\beta,\epsilon}})^2\,du\,,
		\end{align}
		and apply It\^o's Lemma to obtain
		\begin{align}
			\begin{split}
				G_{T} - G_{0} &= \int_{0}^{T}(\partial_{t} + \mathcal{L}^{\nu^{\beta,\epsilon}})\,\widehat{H}_\psi(t,X_{t}^{\nu^{\beta,\epsilon}},Q_{t}^{\nu^{\beta,\epsilon}},P_{t}^{\nu^{\beta,\epsilon}},S_{t};\beta) - \phi\,\,(Q_{t}^{\nu^{\beta,\epsilon}})^2dt\\  
				& \hspace{5mm} + \int_{0}^{T}\sigma\,\partial_{s}\widehat{H}_\psi(t,X_{t}^{\nu^{\beta,\epsilon}},Q_{t}^{\nu^{\beta,\epsilon}},P_{t}^{\nu^{\beta,\epsilon}},S_{t};\beta)\,dW^{s}_{t}\\
				& \hspace{5mm} + \int_{0}^{T}\eta\,\partial_{p}\widehat{H}_\psi(t,X_{t}^{\nu^{\beta,\epsilon}},Q_{t}^{\nu^{\beta,\epsilon}},P_{t}^{\nu^{\beta,\epsilon}},S_{t};\beta)\,dW^{p}_{t}\,,
			\end{split}
		\end{align}
		where the differential operator $\mathcal{L}^{\nu}$ is given in section \ref{sec:performance criterion}. The two stochastic integrands are computed explicitly as
		\begin{align*}
			\partial_{s}\widehat{H}_\psi(t,x,q,p,s;\beta) &= \beta\,\partial_s\gamma_0(t,s)\,q + \beta^2\biggl( \partial_s\lambda_0(t,s) + \partial_s\lambda_1(t,s)\,q + \partial_s\lambda_4(t,s)\,p   \biggr)\,,\\
			\partial_{p}\widehat{H}_\psi(t,x,q,p,s;\beta) &= q + \beta\,\gamma_1(t)\,q + \beta^2\biggl(\lambda_3(t)\,q + \lambda_4(t,s) + 2\,\lambda_5(t)\,p\biggr)\,.
		\end{align*}
		Lemma \ref{lem:lemma2} implies that these stochastic integrands satisfy linear growth conditions, and therefore are square integrable for all admissible controls and the stochastic integrals are martingales. Thus, taking an expectation yields
		\begin{align*}
			\mathbb{E}[G_T] - G_0 &= \mathbb{E}\biggl[\int_{0}^{T}(\partial_{t} + \mathcal{L}^{\nu^{\beta,\epsilon}})\,\widehat{H}_\psi(t,X_{t}^{\nu^{\beta,\epsilon}},Q_{t}^{\nu^{\beta,\epsilon}},P_{t}^{\nu^{\beta,\epsilon}},S_{t};\beta) - \phi\,\,(Q_{t}^{\nu^{\beta,\epsilon}})^2dt\biggr]\,.
		\end{align*}

		Given the explicit form of $\widehat{H}$, we obtain the bound
		\begin{align*}
			\biggl(\partial_{t}+\mathcal{L}^{\nu^{\beta,\epsilon}}\biggr)\,\widehat{H}_\psi(t,x,q,p,s;\beta) - \phi\, q^2 &\le \sup_{\nu}\biggl(\partial_{t}+\mathcal{L}^{\nu}\biggr)\,\widehat{H}_\psi(t,x,q,p,s;\beta) - \phi\,q^2\\
			&=\beta^{3}\,A(t,q,p,s)+\beta^{4}\,B(t,q,p,s)\,,
		\end{align*}
		where the functions $A$ and $B$ are given by
		\begin{align*}
			A(t,q,p,s) &= \frac{1}{2\,k} \,\biggl(\gamma_{0,\psi}(t,s) + \gamma_1(t)\,p + \bigl(b\,\gamma_1(t) + 2\,\gamma_2(t)\bigr)\,q\biggr)\,\biggl(\lambda_1(t,s) + b\,\lambda_4(t,s)  \\
			& \hspace{50mm} + \bigl(\lambda_3(t) + 2\,b\,\lambda_5(t)\bigr)\,p + \bigl(2\,\lambda_2(t) + b\,\lambda_3(t)\bigr)\,q  \biggr)\,,\\
			B(t,q,p,s) &= \frac{1}{4\,k} \,\biggl(\lambda_1(t,s) + b\,\lambda_4(t,s) + \bigl(\lambda_3(t) + 2\,b\,\lambda_5(t)\bigr)\,p + \bigl(2\,\lambda_2(t) + b\,\lambda_3(t)\bigr)\,q\biggr)^2\,.
		\end{align*}
		The aforementioned growth conditions on the functions $\gamma_{0,\psi}$, $\lambda_0$, $\lambda_1$, and $\lambda_4$ imply that the functions $A$ and $B$ satisfy quadratic growth conditions in the variables $q$, $p$, and $s$. Recalling the definition of $G$, this gives
		\begin{align*}
			&\mathbb{E}\biggl[\widehat{H}_\psi(T,X_{T}^{\nu^{\beta,\epsilon}},Q_{T}^{\nu^{\beta,\epsilon}},P_{T}^{\nu^{\beta,\epsilon}},S_{T};\beta) - \int_0^T \phi\,(Q_{t}^{\nu^{\beta,\epsilon}})^2\,dt\biggr] - \widehat{H}_\psi(0,x, q, p, s;\beta) \\ 
			& \hspace{40mm} \leq \beta^3\,\mathbb{E}\biggl[\int_0^T A(t,Q_{t}^{\nu^{\beta,\epsilon}},P_{t}^{\nu^{\beta,\epsilon}},S_{t}) + \beta\,B(t,Q_{t}^{\nu^{\beta,\epsilon}},P_{t}^{\nu^{\beta,\epsilon}},S_{t})\,dt\biggr]\\
			&\mathbb{E}\biggl[X_T^{\nu^{\beta,\epsilon}} + Q_T^{\nu^{\beta,\epsilon}}\,(P_T^{\nu^{\beta,\epsilon}} -\alpha\,Q_T^{\nu^{\beta,\epsilon}}) - \int_0^T \phi\,(Q_{t}^{\nu^{\beta,\epsilon}})^2\,dt\biggr] - \widehat{H}_\psi(0,x, q, p, s;\beta) \\ 
			& \hspace{40mm} \leq \beta^3\,\mathbb{E}\biggl[\int_0^T A(t,Q_{t}^{\nu^{\beta,\epsilon}},P_{t}^{\nu^{\beta,\epsilon}},S_{t}) + \beta\,B(t,Q_{t}^{\nu^{\beta,\epsilon}},P_{t}^{\nu^{\beta,\epsilon}},S_{t})\,dt\biggr]\\
			& H^{\nu^{\beta,\epsilon}}(0,x,q,p,s;\beta) - \widehat{H}_\psi(0,x, q, p, s;\beta) \\ 
			& \hspace{40mm} \leq \beta^3\,\mathbb{E}\biggl[\int_0^T A(t,Q_{t}^{\nu^{\beta,\epsilon}},P_{t}^{\nu^{\beta,\epsilon}},S_{t}) + \beta\,B(t,Q_{t}^{\nu^{\beta,\epsilon}},P_{t}^{\nu^{\beta,\epsilon}},S_{t})\,dt\biggr]\,.
		\end{align*}
		Recalling the definition of $\nu^{\beta,\epsilon}$ gives
		\begin{align*}
			& H_\psi(0,x,q,p,s;\beta) - \widehat{H}_\psi(0,x, q, p, s;\beta) \\ 
			& \hspace{40mm} \leq \epsilon\,\beta^2 + \beta^3\,\mathbb{E}\biggl[\int_0^T A(t,Q_{t}^{\nu^{\beta,\epsilon}},P_{t}^{\nu^{\beta,\epsilon}},S_{t}) + \beta\,B(t,Q_{t}^{\nu^{\beta,\epsilon}},P_{t}^{\nu^{\beta,\epsilon}},S_{t})\,dt\biggr]\,.
		\end{align*}
		By Assumption \ref{ass:ass} ii) and the growth conditions on the functions $A$ and $B$, the expectation is uniformly bounded by a constant $C$ for all sufficiently small $\epsilon$ and $\beta$, giving
		\begin{align*}
			\frac{|H_\psi(0,x,q,p,s;\beta) - \widehat{H}_\psi(0,x, q, p, s;\beta)|}{\beta^2} &\leq \epsilon + \beta\,C\,.
		\end{align*}
		Since $\epsilon>0$ is arbitrary, the desired limit follows. \qed
	\end{proof}

	\begin{proof}[Proof of Theorem \ref{prop:approx_nu}]
		Consider the inventory and perpetual contract price when the agents follows the conjectured approximate strategy, specifically such that
		\begin{align}
			dQ_{t}^{\widehat{\nu}} &= \widehat{\nu}(t,Q^{\widehat{\nu}}_{t},P^{\widehat{\nu}}_{t},S_{t};\beta)\,dt\,,\\
			dP^{\widehat{\nu}}_{t} &= b\,\widehat{\nu}(t,Q^{\widehat{\nu}}_{t},P^{\widehat{\nu}}_{t},S_{t};\beta)\,dt + \eta\,dW_{t}^{p}\,.
		\end{align}
		By Theorem \ref{prop:asymptotic_approximation}, the function $\widehat{\nu}$ may be written as
		\begin{align}
			\widehat{\nu}(t,q,p,s;\beta) &= F_1(t;\beta)\,q + F_2(t;\beta)\,p + \frac{\beta}{2\,k}\,\gamma_{0,\psi}(t,s)\,,
		\end{align}
		where $F_{1}$ and $F_{2}$ are bounded. Therefore $\widehat{\nu}$ is Lipschitz with linear growth in variables $q$, $p$ and $s$ by Lemma \ref{lem:lemma}. Thus, the SDEs for $Q^{\widehat{\nu}}$ and $P^{\widehat{\nu}}$ have a unique strong solution (see Theorem 5.2.9 in \cite{karatzas1991brownian}). Moreover, there exists a constant $\widehat{M}$, such that
		\begin{align}
			\mathbb{E}\biggl[(Q_{t}^{\widehat{\nu}})^2 + (P_{t}^{\widehat{\nu}})^2\biggr] & \leq \widehat{M}\,e^{\widehat{M}\,t}\,, \qquad \forall t\in[0,T]\,.
		\end{align}
		Therefore, by Fubini's Theorem, we have $\mathbb{E}[\int_{0}^{T}\widehat{\nu}_{u}^{2}\,dt]<\infty$ and $\widehat{\nu}$ is an admissible control.
		
		To show that $\widehat{\nu}$ is asymptotically optimal, we proceed with a verification argument while keeping track of the magnitude of the error with respect to optimization, analogous to the proof of Theorem \ref{prop:asymptotic_approximation}. We also remark that with
		\begin{align}
			H_{\psi}(t,x,q,p,s;\beta) &= x + q\,p + h_{\psi}(t,q,p,s;\beta)\,,\\
			H_\psi^{\widehat{\nu}}(t,x,q,p,s;\beta) &= x + q\,p + h_\psi^{\widehat{\nu}}(t,q,p,s;\beta)\,,
		\end{align}
		our desired approximation result is equivalent to
		\begin{align}
			\lim_{\beta\rightarrow 0} \frac{H_{\psi}(t,x,q,p,s;\beta) - H_\psi^{\widehat{\nu}}(t,x,q,p,s;\beta)}{\beta^2} &= 0\,.
		\end{align}
		We prove the accuracy result at $t=0$ with given initial states $x$, $q$, $p$ and $s$, which we henceforth consider to be fixed. The general result for $t\ne 0$ follows similarly. Given the control $\widehat{\nu}$, and the resulting state processes $X^{\widehat{\nu}}$, $Q^{\widehat{\nu}}$, $P^{\widehat{\nu}}$ and $S$, define the process $G = (G_{t})_{t\in[0,T]}$ by
		\begin{align}
			G_t &= X_t^{\widehat{\nu}} + Q_t^{\widehat{\nu}}\,P_t^{\widehat{\nu}} + \widehat{h}_\psi(t,Q_t^{\widehat{\nu}},P_t^{\widehat{\nu}},S_t;\beta) - \int_0^t \phi\,(Q_u^{\widehat{\nu}})^2\,du\,,
		\end{align}
		where $\widehat{h}_\psi$ is the approximation of $h_{\psi}$ given in Theorem \ref{prop:asymptotic_approximation}. Applying It\^o's Lemma to $G$ gives
		\begin{align}
			\begin{split}
				G_{T}-G_{0} =& \int_{0}^{T}(\partial_{t}+\mathcal{L}^{\widehat{\nu}})\,\widehat{H}_\psi(t,X_{t}^{\widehat{\nu}},Q_{t}^{\widehat{\nu}},P_{t}^{\widehat{\nu}},S_{t};\beta) - \phi\,(Q_t^{\widehat{\nu}})^2\,dt\\  
				&+ \int_{0}^{T}\sigma\,\partial_{s}\widehat{H}_\psi(t,X_{t}^{\widehat{\nu}},Q_{t}^{\widehat{\nu}},P_{t}^{\widehat{\nu}},S_{t};\beta)\,dW^{s}_{t}\\
				&+ \int_{0}^{T}\eta\,\partial_{p}\widehat{H}_\psi(t,X_{t}^{\widehat{\nu}},Q_{t}^{\widehat{\nu}},P_{t}^{\widehat{\nu}},S_{t};\beta)\,dW^{p}_{t}\,.
			\end{split}
		\end{align}
		The growth conditions established on the stochastic integrands in the proof of Theorem \ref{prop:asymptotic_approximation} mean that the stochastic integrals are martingales. Thus, we have
		\begin{align}
			\mathbb{E}[G_T] - G_0 &= \mathbb{E}\biggl[\int_{0}^{T}(\partial_{t}+\mathcal{L}^{\widehat{\nu}})\,\widehat{H}_\psi(t,X_{t}^{\widehat{\nu}},Q_{t}^{\widehat{\nu}},P_{t}^{\widehat{\nu}},S_{t};\beta) - \phi\,(Q_t^{\widehat{\nu}})^2\,dt\biggr]\,.
		\end{align}
		By fully expanding the integrand using the expressions in Theorem \ref{prop:asymptotic_approximation} we obtain
		\begin{align}
			(\partial_{t} + \mathcal{L}^{\widehat{\nu}})\,\widehat{H}_\psi(t,X_{t}^{\widehat{\nu}},Q_{t}^{\widehat{\nu}},P_{t}^{\widehat{\nu}},S_{t};\beta) - \phi\,(Q_t^{\widehat{\nu}})^2 = \beta^{3}\,A_{3}(t,Q_{t}^{\widehat{\nu}},P_{t}^{\widehat{\nu}},S_{t})\,,
		\end{align}
		where the function $A_{3}$ is given by
		\begin{align}
			\begin{split}
				A_3(t,q,p,s) &= \frac{1}{2\,k}\biggl(\gamma_{0,\psi}(t,s) + (2\,\gamma_2(t) + b\,\gamma_1(t))\,q + \gamma_1(t)\,p\biggr)\biggl(\lambda_1(t,s) + b\,\lambda_4(t,s)\\
				&\hspace{35mm} + (2\,\lambda_2(t) + b\,\lambda_3(t))\,q + (\lambda_3(t) + 2\,b\,\lambda_5(t))\,p\biggr)\,.
			\end{split}
		\end{align}
		Previously established growth conditions of all terms on the right hand side and the fact that $\widehat{\nu}$ is an admissible control imply that for sufficiently small $\beta$
		\begin{align}
			\beta^3\,\mathbb{E}\biggl[\int_{0}^{T}|A_3(t,Q_{t}^{\widehat{\nu}},P_{t}^{\widehat{\nu}},S_{t})|\,dt\biggr] &\leq \beta^3\,C\,,
		\end{align}
		where $C$ is a finite constant that does not depend on $\beta$. Thus, recalling the definition of $G$ we have
		\begin{align*}
			\biggl|\mathbb{E}\biggl[X_T^{\widehat{\nu}} + Q_T^{\widehat{\nu}}\,P_T^{\widehat{\nu}} + \widehat{h}(T,Q_T^{\widehat{\nu}},P_T^{\widehat{\nu}},S_T;\beta) - \int_0^T \phi\,(Q_t^{\widehat{\nu}})^2\,dt\biggr] - \widehat{H}_\psi(0,x,q,p,s;\beta)\biggr| & \leq \beta^3\,C\\
			\biggl|H_\psi^{\widehat{\nu}}(0,x,q,p,s;\beta) - \widehat{H}_\psi(0,x,q,p,s;\beta)\biggr| & \leq \beta^3\,C\\
			\frac{|H_\psi^{\widehat{\nu}}(0,x,q,p,s;\beta) - \widehat{H}_\psi(0,x,q,p,s;\beta)|}{\beta^2} & \leq \beta\,C\,,
		\end{align*}
		and the desired limit follows. \qed
	\end{proof}
	
	\begin{proof}[Proof of Theorem \ref{prop:T_approx}]
		\underline{Part I} (formal solution): By the terminal condition of the HJB equation \eqref{eqn:HJB_psi}, it is easy to show that $\widetilde{h}_{0}(q) = -\alpha\,q^{2}$. Substituting $\widetilde{h}_\psi$ into the left hand side of \eqref{eqn:HJB_h2} and setting terms proportional to $(T-t)^{0}$ to vanish gives
		\begin{align}
			\widetilde{h}_{1,\psi}(q,p,s) = \biggl(\frac{(b - 2 \alpha)^{2}}{4\,k} - \phi\biggr)\,q^{2} - \beta\,(p - \psi(s))\,q.
		\end{align}
		Similarly, grouping terms proportional to $(T-t)^{1}$ gives
		\begin{align}
			\widetilde{h}_{2,\psi}(q,p,s) &= \frac{b - 2\,\alpha}{4\,k}\biggl(\frac{(b - 2\,\alpha)^{2}}{2\,k} - 2\,\phi - b\,\beta\biggr)\,q^{2} + \frac{\beta}{4}\,\biggl(-\frac{b - 2\,\alpha}{k}\,(p - \psi(s)) + \sigma^{2}\,\psi''(s)\biggr)\,q.
		\end{align}
		\underline{Part II:} (accuracy of approximation).
		With $\widetilde{h}_\psi$ as given in the theorem, define
		\begin{align}
			\widetilde{H}_\psi(t,x,q,p,s;T) &= x + q\,p + \widetilde{h}_\psi(t,q,p,s;T)\,.
		\end{align}
		For simplicity, we prove the approximation holds for $t=0$ with initial states given by $x$, $q$, $p$, and $s$. The case of $t\ne0$ follows similarly. Let $\nu^{T,\epsilon}$ be an admissible control which is $\epsilon\,T^{2}$-optimal. Specifically, the control satisfies
		\begin{align}
			H^{\nu^{T,\epsilon}}(0,x,q,p,s;T) + \epsilon\,T^{2} \ge H_{\psi}(0,x,q,p,s;T)\,.
		\end{align}
		Define the process $G$ by
		\begin{align}
			G_{t} = \widetilde{H}_\psi(t,X_{t}^{\nu^{T,\epsilon}},Q_{t}^{\nu^{T,\epsilon}},P_{t}^{\nu^{T,\epsilon}},S_{t};T) - \int_0^t \phi\,(Q_{u}^{\nu^{T,\epsilon}})^2\,du\,,
		\end{align}
		and apply It\^o's Lemma to obtain
		\begin{align}
			\begin{split}
				G_{T} - G_{0} &= \int_{0}^{T}(\partial_{t} + \mathcal{L}^{\nu^{T,\epsilon}})\,\widetilde{H}_\psi(t,X_{t}^{\nu^{T,\epsilon}},Q_{t}^{\nu^{T,\epsilon}},P_{t}^{\nu^{T,\epsilon}},S_{t};T) - \phi\,\,(Q_{t}^{\nu^{T,\epsilon}})^2dt\\ 
				& \hspace{5mm} + \int_{0}^{T}\sigma\,\partial_{s}\widetilde{H}_\psi(t,X_{t}^{\nu^{T,\epsilon}},Q_{t}^{\nu^{T,\epsilon}},P_{t}^{\nu^{T,\epsilon}},S_{t};T)\,dW^{s}_{t}\\
				& \hspace{5mm} + \int_{0}^{T}\eta\,\partial_{p}\widetilde{H}_\psi(t,X_{t}^{\nu^{T,\epsilon}},Q_{t}^{\nu^{T,\epsilon}},P_{t}^{\nu^{T,\epsilon}},S_{t};T)\,dW^{p}_{t}\,,
			\end{split}
		\end{align}
		where the differential operator $\mathcal{L}^{\nu}$ is given by \eqref{eqn:L_operator}. The two stochastic integrands are computed explicitly as
		\begin{align*}
			\partial_{s}\widetilde{H}_\psi(t,x,q,p,s;T) &= \beta\,\psi'(s)\,q\,(T-t) + \frac{\beta}{4}\,\biggl(\frac{(b - 2\,\alpha)\,\psi'(s)}{k} + \sigma^{2}\psi'''(s)\biggr)\,q\,(T-t)^{2}\,,\\
			\partial_{p}\widetilde{H}_\psi(t,x,q,p,s;T) &= -\beta\,q\,(T-t) - \frac{(b - 2\,\alpha)\,\beta}{4\,k}\,q\,(T-t)^{2}\,.
		\end{align*}
		Boundedness of derivatives of $\psi$ from assumption \ref{ass:ass} implies that these stochastic integrands satisfy linear growth conditions, and therefore are square integrable for all admissible controls and the stochastic integrals are martingales. Thus, taking an expectation yields
		\begin{align*}
			\mathbb{E}[G_T] - G_0 &= \mathbb{E}\biggl[\int_{0}^{T}(\partial_{t} + \mathcal{L}^{\nu^{T,\epsilon}})\,\widetilde{H}_\psi(t,X_{t}^{\nu^{T,\epsilon}},Q_{t}^{\nu^{T,\epsilon}},P_{t}^{\nu^{T,\epsilon}},S_{t};T) - \phi\,\,(Q_{t}^{\nu^{T,\epsilon}})^2dt\biggr]\,.
		\end{align*}
		Given the explicit form of $\widetilde{H}$, we obtain the bound
		\begin{align*}
			&\biggl(\partial_{t}+\mathcal{L}^{\nu^{T,\epsilon}}\biggr)\,\widetilde{H}_\psi(t,x,q,p,s;T) - \phi\, q^2\\
			&\le \sup_{\nu}\biggl(\partial_{t}+\mathcal{L}^{\nu}\biggr)\,\widetilde{H}_\psi(t,x,q,p,s;T) - \phi\,q^2\\
			&=(T-t)^{2}\,\widetilde{A}(q,p,s)+(T-t)^{3}\,\widetilde{B}(q,p,s)+(T-t)^{4}\,\widetilde{C}(q,p,s)\,,
		\end{align*}
		where the functions $\widetilde{A}$, $\widetilde{B}$ and $\widetilde{C}$ are given by
		\begin{align*}
			\widetilde{A}(q,p,s) &= \frac{1}{2}\,\sigma^{2}\,\partial_{ss}\widetilde{h}_{2,\psi} + \frac{1}{4\,k}\,\biggl((b\,\partial_{p}\widetilde{h}_{1,\psi} + \partial_{q}\widetilde{h}_{1,\psi})^{2} + 2\,(b - 2\,\alpha)\,(b\,\partial_{p}\widetilde{h}_{2,\psi} + \partial_{q}\widetilde{h}_{2,\psi})\,q\biggr), \\
			\widetilde{B}(q,p,s) &= \frac{1}{2\,k}\,(b\,\partial_{p}\widetilde{h}_{1,\psi} + \partial_{q}\widetilde{h}_{1,\psi})\,(b\,\partial_{p}\widetilde{h}_{2,\psi} + \partial_{q}\widetilde{h}_{2,\psi}),\\
			\widetilde{C}(q,p,s) &= \frac{1}{4\,k}\,(b\,\partial_{p}\widetilde{h}_{2,\psi} + \partial_{q}\widetilde{h}_{2,\psi})^{2}.
		\end{align*}
		The functions $\widetilde{h}_{1,\psi}$ and $\widetilde{h}_{2,\psi}$ have at most quadratic growth in the variables $q$ and $p$. Substituting the definition of $G$ and applying assumption \ref{ass:ass} gives
		\begin{align*}
			& \biggl|H^{\nu^{T,\epsilon}}(0,x,q,p,s;T) - \widetilde{H}_\psi(0,x, q, p, s;T)\biggr| \leq \mathbb{E}\biggl[\int_0^T \biggl|(T-t)^{2}\,\widetilde{A}(t,Q_{t}^{\nu^{T,\epsilon}},P_{t}^{\nu^{T,\epsilon}},S_{t})\\
			& \hspace{45mm} + (T-t)^{3}\,\widetilde{B}(t,Q_{t}^{\nu^{T,\epsilon}},P_{t}^{\nu^{T,\epsilon}},S_{t}) + (T-t)^{4}\,\widetilde{C}(t,Q_{t}^{\nu^{T,\epsilon}},P_{t}^{\nu^{T,\epsilon}},S_{t})\biggr|\,dt\biggr]\\
			& \hspace{75mm} \leq T^3\,C\,,
		\end{align*}
		for some constant $C$ that does not depend on $T$. Recalling the definition of $\nu^{T,\epsilon}$ gives
		\begin{align*}
			|H_\psi(0,x,q,p,s;T) - \widetilde{H}_\psi(0,x, q, p, s;T)| &\leq \epsilon\,T^{2} + T^{3}\,C\\
			\frac{|H_\psi(0,x,q,p,s;T) - \widetilde{H}_\psi(0,x, q, p, s;T)|}{T^2} &\leq \epsilon + T\,C\,.
		\end{align*}
		Since $\epsilon>0$ is arbitrary, the desired limit follows. \qed
	\end{proof}
	
	\begin{proof}[Proof of Theorem \ref{prop:T_approx_nu}]
		When the agent follows the proposed strategy the inventory and perpetual price processes satisfy
		\begin{align}
			dQ_{t}^{\widetilde{\nu}} &= \widetilde{\nu}(t,Q^{\widetilde{\nu}}_{t},P^{\widetilde{\nu}}_{t},S_{t};T)\,dt\,,\\
			dP^{\widetilde{\nu}}_{t} &= b\,\widetilde{\nu}(t,Q^{\widetilde{\nu}}_{t},P^{\widetilde{\nu}}_{t},S_{t};T)\,dt + \eta\,dW_{t}^{p}\,.
		\end{align}
		By Theorem \ref{prop:T_approx}, the function $\widetilde{\nu}$ may be written as
		\begin{align}
			\widetilde{\nu}(t,q,p,s;T) &= \widetilde{F}_1(T-t)\,q + \widetilde{F}_2(T-t)\,p + \frac{\beta}{2\,k}\,\psi(s)\,,
		\end{align}
		where $\widetilde{F}_{1}$ and $\widetilde{F}_{2}$ are bounded. Therefore $\widetilde{\nu}$ is Lipschitz with linear growth in variables $q$, $p$ and $s$. Thus, the SDEs for $Q^{\widetilde{\nu}}$ and $P^{\widetilde{\nu}}$ have a unique strong solution (see Theorem 5.2.9 in \cite{karatzas1991brownian}). Moreover, there exists a constant $\widetilde{M}$, such that
		\begin{align}
			\mathbb{E}\biggl[(Q_{t}^{\widetilde{\nu}})^2 + (P_{t}^{\widetilde{\nu}})^2\biggr] & \leq \widetilde{M}\,e^{\widetilde{M}\,t}\,, \qquad \forall t\in[0,T]\,.
		\end{align}
		Therefore, by Fubini's Theorem, we have $\mathbb{E}[\int_{0}^{T}\widetilde{\nu}_{u}^{2}\,dt]<\infty$ and $\widetilde{\nu}$ is an admissible control.
		
		To show that $\widetilde{\nu}$ is asymptotically optimal, we proceed with a verification argument while keeping track of the magnitude of the error with respect to optimization, analogous to the proof of Theorem \ref{prop:T_approx}. We also remark that with
		\begin{align}
			H_{\psi}(t,x,q,p,s;T) &= x + q\,p + h_{\psi}(t,q,p,s;T)\,,\\
			H_\psi^{\widetilde{\nu}}(t,x,q,p,s;T) &= x + q\,p + h_\psi^{\widetilde{\nu}}(t,q,p,s;T)\,,
		\end{align}
		our desired approximation result is equivalent to
		\begin{align}
			\lim_{T\rightarrow 0} \frac{H_{\psi}(t,x,q,p,s;T) - H_\psi^{\widetilde{\nu}}(t,x,q,p,s;T)}{T^2} &= 0\,.
		\end{align}
		We prove the accuracy result at $t=0$ with given initial states $x$, $q$, $p$ and $s$, which we henceforth consider to be fixed. The general result for $t\ne 0$ follows similarly. Given the control $\widetilde{\nu}$, and the resulting state processes $X^{\widetilde{\nu}}$, $Q^{\widetilde{\nu}}$, $P^{\widetilde{\nu}}$ and $S$, define the process $G = (G_{t})_{t\in[0,T]}$ by
		\begin{align}
			G_t &= X_t^{\widetilde{\nu}} + Q_t^{\widetilde{\nu}}\,P_t^{\widetilde{\nu}} + \widetilde{h}_\psi(t,Q_t^{\widetilde{\nu}},P_t^{\widetilde{\nu}},S_t;T) - \int_0^t \phi\,(Q_u^{\widetilde{\nu}})^2\,du\,,
		\end{align}
		where $\widetilde{h}_\psi$ is the approximation of $h_{\psi}$ given in Theorem \ref{prop:T_approx}. Applying It\^o's Lemma to $G$ gives
		\begin{align}
			\begin{split}
				G_{T}-G_{0} =& \int_{0}^{T}(\partial_{t}+\mathcal{L}^{\widetilde{\nu}})\,\widetilde{H}_\psi(t,X_{t}^{\widetilde{\nu}},Q_{t}^{\widetilde{\nu}},P_{t}^{\widetilde{\nu}},S_{t};T) - \phi\,(Q_t^{\widetilde{\nu}})^2\,dt\\  
				&+ \int_{0}^{T}\sigma\,\partial_{s}\widetilde{H}_\psi(t,X_{t}^{\widetilde{\nu}},Q_{t}^{\widetilde{\nu}},P_{t}^{\widetilde{\nu}},S_{t};T)\,dW^{s}_{t}\\
				&+ \int_{0}^{T}\eta\,\partial_{p}\widetilde{H}_\psi(t,X_{t}^{\widetilde{\nu}},Q_{t}^{\widetilde{\nu}},P_{t}^{\widetilde{\nu}},S_{t};T)\,dW^{p}_{t}\,.
			\end{split}
		\end{align}
		The growth conditions established on the stochastic integrands in the proof of Theorem \ref{prop:T_approx} mean that the stochastic integrals are martingales. Thus, we have
		\begin{align}
			\mathbb{E}[G_T] - G_0 &= \mathbb{E}\biggl[\int_{0}^{T}(\partial_{t}+\mathcal{L}^{\widetilde{\nu}})\,\widetilde{H}_\psi(t,X_{t}^{\widetilde{\nu}},Q_{t}^{\widetilde{\nu}},P_{t}^{\widetilde{\nu}},S_{t};T) - \phi\,(Q_t^{\widetilde{\nu}})^2\,dt\biggr]\,.
		\end{align}
		By fully expanding the integrand using the expressions in Theorem \ref{prop:T_approx} we obtain
		\begin{align}
			(\partial_{t} + \mathcal{L}^{\widetilde{\nu}})\,\widetilde{H}_\psi(t,X_{t}^{\widetilde{\nu}},Q_{t}^{\widetilde{\nu}},P_{t}^{\widetilde{\nu}},S_{t};T) - \phi\,(Q_t^{\widetilde{\nu}})^2 = (T-t)^{2}\,\widetilde{A}(q,p,s)+(T-t)^{3}\,\widetilde{B}(q,p,s)\,,
		\end{align}
		where the functions $\widetilde{A}$ and $\widetilde{B}$ are given in the proof of Theorem \ref{prop:T_approx}. Previously established growth conditions of all terms on the right hand side and the fact that $\widetilde{\nu}$ is an admissible control imply that for sufficiently small $T$
		\begin{align}
			\mathbb{E}\biggl[\int_{0}^{T}|(T-t)^{2}\,\widetilde{A}(Q_{t}^{\widetilde{\nu}},P_{t}^{\widetilde{\nu}},S_{t}) + (T-t)^{3}\,\widetilde{B}(Q_{t}^{\widetilde{\nu}},P_{t}^{\widetilde{\nu}},S_{t})|\,dt\biggr] &\leq T^3\,C\,(1 + e^{\widetilde{M}\,T})\,,
		\end{align}
		where $C$ is a constant that does not depend on $T$. Thus, recalling the definition of $G$ we have
		\begin{align*}
			\biggl|\mathbb{E}\biggl[X_T^{\widetilde{\nu}} + Q_T^{\widetilde{\nu}}\,P_T^{\widetilde{\nu}} + \widetilde{h}(T,Q_T^{\widetilde{\nu}},P_T^{\widetilde{\nu}},S_T;T) - \int_0^T \phi\,(Q_t^{\widetilde{\nu}})^2\,dt\biggr] - \widetilde{H}_\psi(0,x,q,p,s;T)\biggr| & \leq T^3\,C\,(1 + e^{\widetilde{M}\,T})\\
			\biggl|H_\psi^{\widetilde{\nu}}(0,x,q,p,s;T) - \widetilde{H}_\psi(0,x,q,p,s;T)\biggr| & \leq T^3\,C\,(1 + e^{\widetilde{M}\,T})\\
			\frac{|H_\psi^{\widetilde{\nu}}(0,x,q,p,s;T) - \widetilde{H}_\psi(0,x,q,p,s;T)|}{T^2} & \leq T\,C\,(1 + e^{\widetilde{M}\,T})\,,
		\end{align*}
		and the desired limit follows. \qed
	\end{proof}
	
	\begin{proof}[Proof of Theorem \ref{prop:closed_form_approx}]
		By Theorems \ref{prop:optimal_control} and \ref{prop:T_approx_nu} the controls $\overline{\nu}$ and $\widetilde{\nu}$ are
		\begin{align}
			\overline{\nu}(t,q,p,s;T) &= \nu^{\ast}_{0}(t;T)\,q + \nu^{\ast}_{1}(t;T)\,(p - \psi(s))\,,\\
			\widetilde{\nu}(t,q,p,s;T) &= \widetilde{\nu}_{2}(t;T)\,q + \widetilde{\nu}_{3}(t;T)\,(p - \psi(s))\,,
		\end{align}
		where $\nu^{\ast}_{0}$, $\nu^{\ast}_{1}$, $\widetilde{\nu}_{2}$ and $\widetilde{\nu}_{3}$ are given by
		\begin{align}
			\nu^{\ast}_{0}(t;T) &= \frac{1}{4\,k}\,\left((\xi(t;T)+\pi(t;T)\right)\,,\\
			\nu^{\ast}_{1}(t;T) &= \frac{1}{4\,k\,b}\,\left((\xi(t;T)-\pi(t;T)\right)\,,\\
			\widetilde{\nu}_{2}(t;T) &= \frac{1}{2\,k}\,\biggl[b - 2\,\alpha + \biggl(\frac{1}{2\,k}\,(b - 2\,\alpha)^{2} - 2\,\phi - b\,\beta\biggr)\,(T-t)\biggr]\,,\\
			\widetilde{\nu}_{3}(t;T) &= -\frac{1}{2\,k}\,\beta\,(T-t)\,,
		\end{align}
		where $\xi$ and $\pi$ are given in Theorem \ref{prop:optimal_control}. Showing that $\overline{\nu}$ is admissible follows the same reasoning as showing $\widetilde{\nu}$ is admissible from the proof of Theorem \ref{prop:T_approx_nu}. A direct computation gives
		\begin{align*}
			\lim_{T \to 0}\xi(t;T) = \lim_{T \to 0}\pi(t;T) = b - 2\,\alpha\,,\\
		\end{align*}
		which further gives
		\begin{align*}
			\lim_{T \to 0}(\nu^{\ast}_{0}(t;T) - \widetilde{\nu}_{2}(t;T)) &= 0\,,\\
			\lim_{T \to 0}(\nu^{\ast}_{1}(t;T) - \widetilde{\nu}_{3}(t;T)) &= 0\,.
		\end{align*}
		The first derivative of $\xi$ and $\pi$ with respect to $T$ is computed as
		\begin{align*}
			\partial_{T}\xi(t;T) &= -4\,a\,\omega\,C\,\frac{e^{-2\,\omega\,(T-t)}}{(C\,e^{-2\,\omega\,(T-t)} + 1)^{2}}\,,\\
			\partial_{T}\pi(t;T) &= (C + 1)\,(b - 2\,\alpha)\,\frac{2\,\omega\,C\,e^{-3\,\omega\,(T-t)} - \omega\,e^{-\omega\,(T-t)}\,(C\,e^{-2\,\omega\,(T-t)} + 1)}{(C\,e^{-2\,\omega\,(T-t)} + 1)^{2}}\\
			& - \frac{8\,k\,\phi\,\omega\,C\,e^{-2\,\omega\,(T-t)}\,(C\,e^{-\omega\,(T-t)} + 1)\,(1 - e^{-\omega\,(T-t)})}{a\,(C\,e^{-2\,\omega\,(T-t)} + 1)^{2}}\\
			&-\frac{4\,k\,\phi\,(C\,e^{-2\,\omega\,(T-t)} + 1)\,(\omega\,e^{-\omega\,(T-t)}\,(C\,e^{-\omega\,(T-t)} + 1) - \omega\,C\,e^{-\omega\,(T-t)}\,(1 - e^{-\omega\,(T-t)}))}{a\,(C\,e^{-2\,\omega\,(T-t)} + 1)^{2}}\,,\\
		\end{align*}
		where the constants $a$, $C$ and $\omega$ are stated in Theorem \ref{prop:optimal_control}. A tedious but direct computation yields
		\begin{align*}
			\lim_{T \to 0}\partial_{T}\xi(t;T) &= \frac{1}{2\,k}\,(b - 2\,\alpha)^{2} - 2\,(b\,\beta + \phi)\,,\\
			\lim_{T \to 0}\partial_{T}\pi(t;T) &= \frac{1}{2\,k}\,(b - 2\,\alpha)^{2} - 2\,\phi\,.
		\end{align*}
		Hence
		\begin{align*}
			\lim_{T \to 0}\partial_{T}(\nu^{\ast}_{0}(t;T) - \widetilde{\nu}_{2}(t;T)) &= 0\,,\\
			\lim_{T \to 0}\partial_{T}(\nu^{\ast}_{1}(t;T) - \widetilde{\nu}_{3}(t;T)) &= 0\,.
		\end{align*}
		Combining all the limits which are given above implies that the following limit holds locally uniformly in $(t,q,p,s)$ by L'Hopital's rule:
		\begin{align*}
			\lim_{T\to 0}\frac{\overline{\nu}(t,q,p,s;T) - \widetilde{\nu}(t,q,p,s;T)}{T} = 0\,.
		\end{align*}
		Given the candidate strategy $\overline{\nu}_{t} = \nu^{\ast}(t,Q^{\overline{\nu}}_{t},P^{\overline{\nu}}_{t},\psi(S_{t});T)$, define the stochastic process $(G_{t})_{t\in [0,T]}$ by
		\begin{align*}
			G_{t} = \widetilde{H}_\psi(t,X_{t}^{\overline{\nu}},Q_{t}^{\overline{\nu}},P_{t}^{\overline{\nu}},S_{t};T) - \int_0^t \phi\,(Q_{u}^{\overline{\nu}})^2\,du\,,
		\end{align*}
		and $\widetilde{H}_{\psi}$ is the approximation of $H_{\psi}$ in Theorem \ref{prop:T_approx}. Apply Ito's Lemma to $G$ and write
		\begin{align*}
			\begin{split}
				G_{T} - G_{0} &= \int_{0}^{T}(\partial_{t} + \mathcal{L}^{\overline{\nu}})\,\widetilde{H}_\psi(t,X_{t}^{\overline{\nu}},Q_{t}^{\overline{\nu}},P_{t}^{\overline{\nu}},S_{t};T) - \phi\,\,(Q_{t}^{\overline{\nu}})^2dt\\ 
				& \hspace{5mm} + \int_{0}^{T}\sigma\,\partial_{s}\widetilde{H}_\psi(t,X_{t}^{\overline{\nu}},Q_{t}^{\overline{\nu}},P_{t}^{\overline{\nu}},S_{t};T)\,dW^{s}_{t}\\
				& \hspace{5mm} + \int_{0}^{T}\eta\,\partial_{p}\widetilde{H}_\psi(t,X_{t}^{\overline{\nu}},Q_{t}^{\overline{\nu}},P_{t}^{\overline{\nu}},S_{t};T)\,dW^{p}_{t}\,.
			\end{split}
		\end{align*}
		The growth conditions established on the stochastic integrands in the proof of Theorem \ref{prop:T_approx} mean that the stochastic integrals are martingales. Defining $r(t,q,p,s;T) = \overline{\nu}(t,q,p,s;T) - \widetilde{\nu}(t,q,p,s;T)$
		\begin{align*}
			(\partial_{t} + \mathcal{L}^{\overline{\nu}})\,\widetilde{H}_\psi - \phi\,q^{2} &= \partial_{t}\widetilde{h}_{\psi} + (\partial_{q}\widetilde{h}_{\psi} + b(q + \partial_{p}\widetilde{h}_{\psi}))\,\overline{\nu} - k\,(\overline{\nu})^{2}\\
			& \hspace{15mm} - \beta\,q\,(p - \psi(s)) + \frac{1}{2}\,\sigma^{2}\,\partial_{ss}\widetilde{h}_{\psi} - \phi\,q^{2}\\
			&= (\partial_{t} + \mathcal{L}^{\widetilde{\nu}})\,\widetilde{H}_\psi  - \phi\,q^{2}+ (\partial_{q}\widetilde{h}_{\psi} + b\,(q + \partial_{p}\widetilde{h}_{\psi}) - 2\,k\,\widetilde{\nu})\,r - k\,r^{2}\\
			&= (T-t)^{2}\,\widetilde{A}(q,p,s) + (T-t)^{3}\,\widetilde{B}(q,p,s) + V(t,q,p,s;T),
		\end{align*}
		where the function $V$ is given by
		\begin{align*}
			V(t,q,p,s;T) &= (\partial_{q}\widetilde{h}_{\psi} + b\,(q + \partial_{p}\widetilde{h}_{\psi}) - 2\,k\,\widetilde{\nu})\,r - k\,r^{2}\\
			&= (T-t)^{2}\,(\partial_{q}\widetilde{h}_{2,\psi} + b\,\partial_{p}\widetilde{h}_{2,\psi})\,r - k\,r^{2}.
		\end{align*}
		Since the functions $V(t,q,p,s;T)$ is at most quadratic growth in variables $q$ and $p$ and we have already shown that $r(t,q,p,s;T) = o(T)$ as $T\to 0$, we have
		\begin{align*}
			\lim_{T\to 0}\frac{V(t,q,p,s;T)}{T^{2}} = 0.
		\end{align*}
		Taking an expectation and combining all the results yields
		\begin{align*}
			\biggl|\mathbb{E}[G_T] - G_0\biggr| &= \biggl|\mathbb{E}\biggl[\int_{0}^{T}(\partial_{t}+\mathcal{L}^{\overline{\nu}})\,\widetilde{H}_\psi(t,X_{t}^{\overline{\nu}},Q_{t}^{\overline{\nu}},P_{t}^{\overline{\nu}},S_{t};T) - \phi\,(Q_t^{\overline{\nu}})^2\,dt\biggr]\biggr|\\
			&\le T^3\,C\,(1 + e^{\widetilde{M}\,T}) + V(T).
		\end{align*}
		where the function $V(T)$ can be chosen to satisfy
		\begin{align*}
			\biggl|\mathbb{E}\biggl[\int_{0}^{T}V(t,Q_{t}^{\overline{\nu}},P_{t}^{\overline{\nu}},S_{t};T)\,dt\biggr]\biggr| \le V(T),
		\end{align*}
		and
		\begin{align*}
			\lim_{T\to 0}\frac{V(T)}{T^{3}} = 0.
		\end{align*}
		Thus, recalling the definition of $G$ we have
		\begin{align*}
			\biggl|H_\psi^{\overline{\nu}}(0,x,q,p,s;T) - \widetilde{H}_\psi(0,x,q,p,s;T)\biggr| & \leq T^3\,C\,(1 + e^{\widetilde{M}\,T}) + V(T)\\
			\frac{|H_\psi^{\overline{\nu}}(0,x,q,p,s;T) - \widetilde{H}_\psi(0,x,q,p,s;T)|}{T^2} & \leq T\,C\,(1 + e^{\widetilde{M}\,T}) + \frac{V(T)}{T^{2}}\,,
		\end{align*}
		and the desired limit follows. \qed
	\end{proof}

	\bibliographystyle{chicago}
	\bibliography{References}

\end{document}